\documentclass[11pt,reqno]{amsart}
\usepackage[T1]{fontenc}
\usepackage{graphicx}
\usepackage{amssymb,amsmath,mathrsfs,bm,braket,marginnote}
\usepackage{enumerate}
\usepackage[colorlinks=true, pdfstartview=FitV, linkcolor=blue, citecolor=blue, urlcolor=blue]{hyperref}
\usepackage{tikz}

\numberwithin{equation}{section}
\usepackage{amsthm}

\newtheorem{theorem}{Theorem}[section]
\newtheorem{lemma}[theorem]{Lemma}
\newtheorem{definition}[theorem]{Definition}
\newtheorem{remark}[theorem]{Remark}
\newtheorem{coro}[theorem]{Corollary}
\newtheorem{example}[theorem]{Example}

\newtheorem{proposition}[theorem]{Proposition}

\reversemarginpar

\def\J{\mathrm{J}}
\def\d{\mathrm{d}}
\def\sgn{\mathrm{sgn}}
\def\Tr{\mathrm{Tr}}
\def\DO{\mathrm{DO}_{N,M}}
\def\NDDO{\mathrm{NDDO}_{N,M}}

\def\D{\mathcal{D}}

\def\mat{\mathrm{Mat}_\infty(\mathbb{R})}
\def\Ker{\mathrm{Ker}\,}

\author[B. Wang]{Bao Wang}
\address{School of Mathematics and Statistics, Ningbo University, Ningbo 315211, PR China.}
\email{wangbao@nbu.edu.cn}

\begin{document}

\title[Integrable pentagram-type maps on polyhedra via  partial difference operators]{Integrable pentagram-type maps on polyhedra via  partial difference operators}


\begin{abstract}
This paper introduces a family of natural generalizations of the pentagram map from polygons to (twisted) polyhedra and proves their integrability  through the  partial difference operators. 
A canonical special case, 
which corresponds to the discrete Laplace transformation of discrete conjugate nets, 
is investigated in detail.
We first establish a canonical bijection between the projective equivalence classes of these polyhedra in $\mathbb{RP}^3$ and the  spectral data of doubly periodic partial difference operators modulo the gauge actions.
Furthermore, we prove the complete integrability of these pentagram-type maps by explicitly identifying them with the refactorization maps on the Poisson-Lie group of pseudo partial difference operators. 
This algebraic identification naturally yields an explicit Lax representation and an $r$-matrix induced Poisson bracket for the geometric dynamics.
\end{abstract}
\maketitle
\section{Introduction}
The pentagram map, originally introduced by Schwartz \cite{Schwartz1992The} in 1992,
has now become a popular discrete system.
It is defined for polygons in the projective plane and has an elegant geometric meaning:
it maps a vertex $v_k$ to the intersection of two segments $\overline{v_{k-1}v_{k+1}}$ and $\overline{v_kv_{k+2}}$.
The pentagram map has attracted  considerable attention, mainly due to its deep connections to integrable systems \cite{2012Integrability,2010The,2013Liouville,Soloviev2013Integrability}, cluster algebras \cite{2014Integrable,GLICK20111019,glick2016}, dimer models \cite{dimer2022}, Poisson-Lie groups \cite{refactorization2020}, and Poncelet polygons \cite{izosimov2022pentagram}.

So far,
there have been a variety of generalizations of the pentagram map,
such as the high dimensional pentagram maps on so-called corrugated polygons \cite{2014Integrable},
the short-diagonal pentagram maps \cite{2012Integrability,G2014On},
the dented pentagram maps \cite{dented},
the long-diagonal pentagram maps \cite{long2022},
generalized pentagram maps from $Y$-meshes \cite{glick2016},
the pentagram maps on Grassmannians \cite{grassmann,Non},
the pentagram-type maps related to $Q$-nets \cite{q_net2026},
the coupled pentagram maps \cite{coupled_sub},
the noncommutative leapfrog maps \cite{wang2025noncommutative},
and the non-integrable pentagram maps \cite{nonintegrable2015}.

A natural question is how to generalize such integrable dynamics from polygons to polyhedra.
Indeed, studying integrability of the  pentagram-type maps on polyhedra and extending the refactorization framework from ordinary difference operators to partial difference operators  were explicitly posed as  open problems by Izosimov \cite{refactorization2020}.
In this paper, we answer this question by studying  pentagram-type maps on a special class of polyhedra in $\mathbb{RP}^d$,
termed  twisted $\J$-corrugated $(N, M)$-hedra, which combine local $\J$-corrugation constraints with  twisted conditions (see Definitions \ref{def_j} and \ref{def_twi} for precise formulations).
To provide a concrete geometric picture, it is  instructive to look at a canonical special case of the generalized maps: the discrete Laplace transformation of discrete conjugate nets,
which was originally introduced by Doliwa  \cite{doliwa1997geometric} in the context of discrete differential geometry.
Let $\{v_{i,j}\}_{i,j\in\mathbb{Z}}$ be a discrete conjugate net in $\mathbb{RP}^3$, i.e., all elementary quadrilaterals $\{v_{i,j}, v_{i+1,j}, v_{i,j+1}, v_{i+1,j+1}\}$ are coplanar  (Such nets are essentially $\J$-corrugated polyhedra in $\mathbb{RP}^3$ with $\J = \{(0,0), (0,1), (1,0), (1,1)\}$).
The discrete Laplace transform is a map $T:v\rightarrow \tilde{v}$,
where
\begin{align}\label{def_lap}
	\tilde{v}_{i,j}=(v_{i,j},v_{i,j+1})\cap(v_{i+1,j},v_{i+1,j+1}),
\end{align}
where $(v,u)$ represents the line passing through $v$ and $u$
(see Figure \ref{fig_laplace}).
\begin{figure}[htbp]\label{fig_laplace}
	\centering
	\includegraphics[trim=8cm 5.5cm 6cm 9.4cm, clip, width=0.6\textwidth]{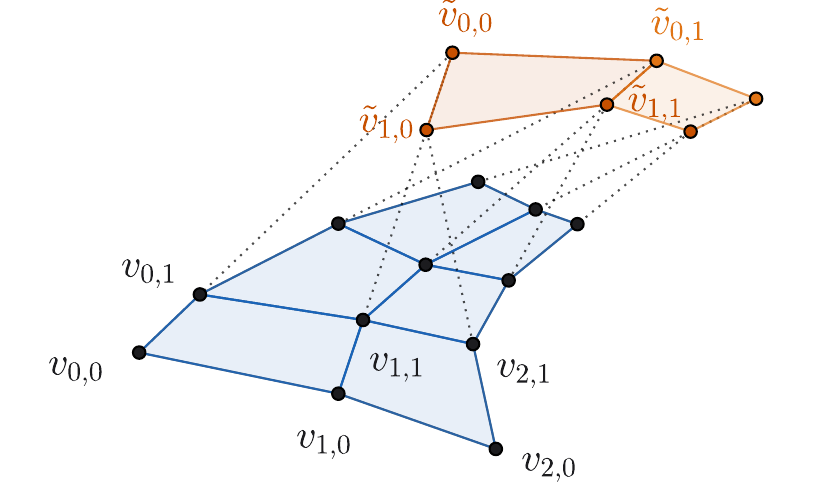}
	\caption{Discrete Laplace transform on a discrete conjugate net}
\end{figure}

Our first main result is establishing  a canonical one-to-one correspondence between the projective equivalence classes of generic twisted $\J$-corrugated polyhedra and the spectral data of doubly periodic partial difference operators. 
\begin{theorem}
	Let $\J = \{(0,0), (0,1), (1,0), (1,1)\}$. 
	There is a bijection between the projective equivalence classes of generic twisted $\J$-corrugated $(N, M)$-hedra in $\mathbb{RP}^3$ and the gauge-equivalence classes of spectral data $(\mathcal{D},  \sum_{k=1}^4 Q_k)$, where $\mathcal{D}$ is a non-degenerate $(N,M)$-periodic partial difference operator and $\sum_{k=1}^4 Q_k$ is a real divisor consisting of four distinct marked points on the spectral curve $\Gamma_{\mathcal{D}}$, invariant under the complex conjugation $\tau$ (i.e., $\tau(\sum_{k=1}^4 Q_k) = \sum_{k=1}^4 Q_k$).
\end{theorem}
Second, we endow the infinite-dimensional group of pseudo partial difference operators with a  Poisson-Lie group structure. 
\begin{theorem}
	The group consisting of $(N,M)$-periodic invertible pseudo partial difference operators is a Poisson-Lie group. 
	The compatible Poisson bracket is given by the Sklyanin bracket:
	$$ \{f_1, f_2\} = \frac{1}{2} \left( \langle r(\nabla^r f_1), \nabla^r f_2 \rangle - \langle r(\nabla^l f_1), \nabla^l f_2 \rangle \right), $$
	where $r = p_{>0} - p_{<0}$ is the classical $r$-matrix associated with the lexicographical decomposition of the corresponding Lie algebra. 
\end{theorem}
Notably, when reduced to ordinary difference operators, this Poisson-Lie structure recovers the one-dimensional results established in \cite{refactorization2020}.
This rigorous algebraic framework allows us to get the integrability of the geometric dynamics.
\begin{theorem}[Integrability via refactorization]
	The  pentagram-type map \eqref{def_lap} on twisted $\J$-corrugated $(N,M)$-hedra is  integrable. 
	Geometrically, the map acts on the space of pairs $(\mathcal{D}, \sum_{k=1}^4 Q_k)$. Under this map, the  divisor $\sum_{k=1}^4 Q_k$ on the spectral curve remains strictly invariant modulo the gauge actions, while the dynamics of the operator $\mathcal{D} = \mathcal{D}_{+} + \mathcal{D}_{-}$ exactly manifest as the refactorization dynamics:
	\begin{align*}
	\tilde{\mathcal{D}}_+ \mathcal{D}_- = \tilde{\mathcal{D}}_- \mathcal{D}_+ 
	\end{align*}
	on the Poisson-Lie group of pseudo partial difference operators. This algebraic identification inherently guarantees the spectral invariance and naturally yields the Lax representation $\mathcal{L}:=\D_-^{-1}\D_+\mapsto\tilde{\mathcal{L}} =: \mathcal{D}_+ \mathcal{L} \mathcal{D}_+^{-1}$, along with an $r$-matrix induced Sklyanin Poisson bracket.
\end{theorem}
Finally, by introducing the cross-ratio variables $y_{i,j}$, we demonstrate that the geometric evolution \eqref{def_lap} transforms into
\begin{align*}
	y_{i+2,j+1}\tilde{\tilde{y}}_{i+1,j}=\frac{(1+\tilde{y}_{i+1,j})(1+\tilde{y}_{i+2,j+1})}{\left(1+\tilde{y}_{i+2,j}^{-1}\right)\left(1+\tilde{y}_{i+1,j+1}^{-1}\right)},
\end{align*}
which exhibits a connection to the
Y-mutations in a certain cluster algebra.

The remainder of the paper is organized as follows. In Section \ref{sec_j}, we introduce the geometric setting of twisted $\J$-corrugated $(N,M)$-hedra in the projective space, and define the pentagram-type maps on them. 
Section \ref{sec_partial} sets up the rigorous algebraic framework of doubly periodic partial difference operators. 
Besides,
we establish a canonical bijection between the projective equivalence classes of the polyhedra and the  spectral data of partial difference operators  modulo the gauge actions. 
In Section \ref{sec_poi}, we explore the Poisson-Lie group structure of pseudo partial difference operators, 
explicitly derive the Poisson brackets for coordinate functions, 
and discuss some reductions of partial difference operators. 
Section \ref{sec_refac} proves that the  pentagram-type maps exactly manifest as integrable refactorization maps within the Poisson-Lie group of pseudo partial difference operators.
Besides, 
we provide the explicit coordinate evolution for \eqref{def_lap} and transform them into a canonical Y-system. 
Finally, Section \ref{sec_con} concludes the paper with some open problems and indications of
future work.

\section{Pentagram-type maps on J-corrugated polyhedra }\label{sec_j}

\begin{definition}
	A polyhedron in $\mathbb{RP}^d$ is a map $v:\mathbb{Z}^2\rightarrow\mathbb{RP}^d$, where the image of the integer lattice forms the vertex set of a quadrilateral mesh.
\end{definition}

Similar to the $\J$-corrugated polygons defined in \cite{refactorization2020},
we define a restricted class of polyhedra whose vertices satisfy certain additional coplanarity conditions.
\begin{definition}\label{def_j}
Let $\J\subset\mathbb{Z}^2$, $|\J|\geq 3$ be a finite set,
and let $d:=|\J|-1$.
Then a polyhedron $\{v_{i,j}\}$ in $\mathbb{RP}^d$ is called $\J$-corrugated if for any $(i,j)\in\mathbb{Z}^2$,
the points $\{v_{i+n,j+m}\mid(n,m)\in\J\}$ belong to a $(|\J|-2)$-dimensional projective subspace.
\end{definition}

	In order to establish a rigorous algebraic correspondence, 
	we impose a natural geometric non-degeneracy condition on these polyhedra. 
	We say a $\J$-corrugated polyhedron is \textbf{ non-degenerate} if for any $(i,j) \in \mathbb{Z}^2$, 
	the vertices $\{v_{i+n,j+m}\mid(n,m)\in\J\}$ form a proper $(|\J|-2)$-dimensional projective subspace, meaning no $|\J|-1$ of these vertices belong to a $(|\J|-3)$-dimensional projective subspace.

\begin{example}\label{exam_corrugated}
Let $\J=\{(0,0),\,(0,1),\,(1,0),\,(1,1)\}$,
then a $\J$-corrugated polyhedron in $\mathbb{RP}^3$ is defined as a map $v:\mathbb{Z}^2\rightarrow\mathbb{RP}^3$ such that all the elementary quadrilaterals
\begin{align*}
v_{i,j},\quad
v_{i,j+1},\quad
v_{i+1,j},\quad
v_{i+1,j+1}
\end{align*}
are planar.
This case corresponds to discrete conjugate nets \cite{doliwa1997geometric}, which were originally studied in three-dimensional Euclidean space $\mathbb{E}^3$.

\end{example}

Now we introduce a natural geometric evolution on the space of twisted $\J$-corrugated polyhedra. 
Suppose that $\J\subset\mathbb{Z}^2$ can be partitioned as $\J=\J_+\cup\J_-$,
where $\J_\pm$ have the same common difference.
For example,
if $\J=\{(0,0),\,(0,1),\,(1,0),\,(1,1)\}$, 
we can define  $\J_+:=\{(0,0),(0,1)\}$ and $\J_-:=\{(1,0),(1,1)\}$,
both of which are arithmetic progressions with the same common difference $(0,1)$.
\begin{definition}
Let $\J_\pm\subset\mathbb{Z}^2$ be non-empty disjoint finite sets with the same common difference.
The pentagram-type map associated with the pair $\J_\pm$ is the map from the space of $\J$-corrugated polyhedra to itself defined by
\begin{align}\label{eq_pen}
\tilde{v}_{i,j}:=\operatorname{span}(v_{i+n,j+m}\mid(n,m)\in\J_+)\cap\operatorname{span}(v_{i+n,j+m}\mid(n,m)\in\J_-),
\end{align}
where $v$ and $\tilde{v}$ are the initial polyhedron and its image under the map,
and $\operatorname{span}(v_{i,j})$ denotes the projective subspace spanned by the points $v_{i,j}$.
\end{definition}

\begin{proposition}
The pentagram-type map \eqref{eq_pen} is a generically well-defined mapping from the space of $\J$-corrugated polyhedra to itself.
\end{proposition}
\begin{proof}
According to Definition \ref{def_j},
for a generic $\J$-corrugated polyhedron $(v)$ in $\mathbb{RP}^d$ with $d=|\J|-1$,
the dimension of the subspace spanned by the points $\{v_{i+n,j+m}\mid(n,m)\in\J\}$  is $|\J|-2$.
Besides,
the dimension of $\operatorname{span}(v_{i+n,j+m}\mid(n,m)\in\J_+)$ is $|\J_+|-1$
and the dimension of $\operatorname{span}(v_{i+n,j+m}\mid(n,m)\in\J_-)$ is $|\J_-|-1$.
Generically, the intersection of these two subspaces defines a unique point $\tilde{v}_{i,j}\in\mathbb{RP}^d$.
Then we just need to prove the polyhedron $\tilde{v}$ is $\J$-corrugated.
Consider the space $L_{i,j}:=\operatorname{span}\{v_{i+n,j+m}\mid(n,m)\in\J_++\J_-\}$,
where $\J_++\J_-$ is the Minkowski sum of $\J_+$ and $\J_-$.
Then we have
\begin{align*}
\operatorname{span}\{\tilde{v}_{i+n,j+m}\mid(n,m)\in\J\}\subset L_{i,j}.
\end{align*}
Indeed, if $(n,m)\in\J_+$,
we have $\tilde{v}_{i+n,j+m}\in\operatorname{span}\{v_{i+n+n',j+m+m'}\mid(n',m')\in\J_-\}$.
Since the latter is a subspace of $L_{i,j}$,
we have $\tilde{v}_{i+n,j+m}\in L_{i,j}$ for all $(n,m)\in\J_+$.
Similarly, 
we can prove $\tilde{v}_{i+n,j+m}\in L_{i,j}$ for all $(n,m)\in\J_-$,
and hence  \(\operatorname{span}\{\tilde{v}_{i+n,j+m}\mid(n,m)\in\J\}\subset L_{i,j}.\)
Since there are $|\J_+|+|\J_-|-1$ elements in $|\J_++\J_-|$,
the dimension of $L_{i,j}$ is $|\J|-2$,
which means the polyhedron $(\tilde{v})$ is $\J$-corrugated.
\end{proof}
\begin{example}
As in Example \ref{exam_corrugated},
let $\J=\{(0,0),\,(0,1),\,(1,0),\,(1,1)\}$.
We can define  $\J_+:=\{(0,0),(0,1)\}$ and $\J_-:=\{(1,0),(1,1)\}$ with the common difference $(0,1)$.
The pentagram-type map is
\begin{align*}
\tilde{v}_{i,j}=(v_{i,j},v_{i,j+1})\cap(v_{i+1,j},v_{i+1,j+1}).
\end{align*}
Similarly, taking $\J_+:=\{(0,0),(1,0)\}$ and $\J_-:=\{(0,1),(1,1)\}$,
it yields the map
\begin{align*}
	\tilde{v}_{i,j}=(v_{i,j},v_{i+1,j})\cap(v_{i,j+1},v_{i+1,j+1}).
\end{align*}
The two maps are corresponding to the discrete Laplace transformations defined on the discrete conjugate nets \cite{doliwa1997geometric}.
\end{example}

In this paper, we primarily focus on a special class of polyhedra, which is introduced in the following definition.
\begin{definition}\label{def_twi}
	A twisted $(N,M)$-hedron is a polyhedron $\{v_{i,j}\in\mathbb{RP}^d\}$ such that $v_{i+N,j}=\phi_1(v_{i,j})$ and $v_{i,j+M}=\phi_2(v_{i,j})$ for all $(i,j)\in\mathbb{Z}^2$ and 
	two fixed (not depending on $(i,j)$) commuting projective transformations $\phi_1$, $\phi_2:\mathbb{RP}^d\rightarrow\mathbb{RP}^d$,
	called the monodromies.
\end{definition}

\section{Partial difference operators}\label{sec_partial}
Having established the geometric framework of twisted $\J$-corrugated $(N,M)$-hedra in the previous section, we now turn our attention to their algebraic counterpart. 
In this section, we introduce the  partial difference operators, which will serve as the primary algebraic tool for describing the associated pentagram-type dynamics.
\subsection{Partial difference operators}
Let $\mat$ be the space of bi-infinite matrices of real numbers,
and let $\J\subset\mathbb{Z}^2$ be a finite collection.
A linear operator $\D:\mat\rightarrow\mat$ is called a partial difference operator supported in $\J$ if it can be written as 
\begin{align*}
(\D X)_{i,j}=\sum_{(n,m)\in\J}a_{i,j}^{n,m}X_{i+n,j+m},
\end{align*}
or equivalently,
\begin{align*}
\D=\sum_{(n,m)\in\J}a^{n,m}T_1^nT_2^m,
\end{align*}
where $X\in\mat$, $T_1:\mat\rightarrow\mat$ is the shift operator $(T_1X)_{i,j}=X_{i+1,j}$ and similarly $(T_2X)_{i,j}=X_{i,j+1}$, and each coefficient $a^{n,m}$ is a bi-infinite matrix acting on $\mat$ by term-wise multiplication,
which can also be regarded as partial difference operators with $\J=\{(0,0)\}$.

Now we introduce a total ordering on  the index set $\mathbb{Z}^2$.

\begin{definition}[Lexicographical order on $\mathbb{Z}^2$]
	For any two indices $(n, m), (n', m') \in \mathbb{Z}^2$, we define the lexicographical order $\prec$ such that $(n, m) \prec (n', m')$ if and only if $n < n'$, or
	$n = n'$ and $m < m'$.
\end{definition}
A partial difference operator $\D = \sum_{(n,m)\in J} a^{n,m}T_1^n T_2^m$ is $(N,M)$-periodic if all its coefficients $a^{m,n}$ are $(N,M)$-periodic bi-infinite matrices, that is, $a^{m,n}_{i+N,j}=a^{m,n}_{i,j+M}=a^{m,n}_{i,j}$.
This is equivalent to saying that $\D$ commutes with the $(N,M)$-power of the shift operators,
that is
$\D T_1^N=T_1^N\D$ and $\D T_2^M=T_2^M\D$.
A partial difference operator $\mathcal{D}$ supported in $\J$ is said to be \textbf{ non-degenerate} if none of its  coefficients $a^{n,m}$ vanish anywhere. 
That is, $a_{i,j}^{n,m} \neq 0$ for all $(n,m) \in \J$ and all $(i,j) \in \mathbb{Z}^2$. 

We denote the space of $(N,M)$-periodic partial difference operators supported in $\J$ by $\DO(\J)$.
Besides, $\NDDO(\J)\subset\DO(\J)$ stands for the (dense) subset of non-degenerate operators,
which will be used to perfectly match the geometric non-degeneracy.
Let $\DO$ be the algebra of all $(N,M)$-periodic partial difference operators with arbitrary finite support.

\subsection{Partial difference operators and J-corrugated polyhedra}
There is a close relation between partial difference operators supported in $\J$ and $\J$-corrugated polyhedra.
Denote by $\mathcal{P}_{N,M}^d(\J)$ the space of twisted $\J$-corrugated $(N,M)$-hedra in $\mathbb{RP}^d$ modulo projective transformations.

Let $H$ be the group of non-vanishing $(N,M)$-quasi-periodic bi-infinite matrices, i.e.
\begin{align*}
H:=\{&\alpha\in\mat\mid\forall (i,j)\in\mathbb{Z}^2,\,\alpha_{i,j}\neq 0, \quad\text{and}\\
&\exists z_1,\,z_2\in\mathbb{R}^*,\,s.t.\,\alpha_{i+N,j}=z_1\alpha_{i,j},\,\alpha_{i,j+M}=z_2\alpha_{i,j}\}.
\end{align*}
Further,
let $H\tilde{\times}H$ be the subgroup of $H\times H$ that consists of pairs of non-vanishing $(N,M)$-quasi-periodic bi-infinite matrices with the same monodromies, i.e.
\begin{align}\label{eq_h}
\begin{split}
H\tilde{\times}H:=
\{
(\alpha,\beta)\in H\times H\mid\exists z_1,\,z_2\in\mathbb{R}^*,\,s.t.\,\alpha_{i+N,j}=z_1\alpha_{i,j},\,\alpha_{i,j+M}=z_2\alpha_{i,j},\\
\beta_{i+N,j}=z_1\beta_{i,j},\,\beta_{i,j+M}=z_2\beta_{i,j}
\}.
\end{split}
\end{align}
The group acts on $\DO(\J)$ by means of the left-right action
\begin{align}\label{eq_lr}
\D\rightarrow\beta^{-1}\D\alpha.
\end{align}
The Floquet-Bloch theory for partial difference operators has been considered in \cite{Krichever1985,Oblo2000}.
Recall that a Bloch solution $\psi(w_1,w_2)$ of $\D$ associated with a point $(w_1,w_2)\in(\mathbb{C}^*)^2$ satisfies:
\begin{align}\label{eq_omega}
\psi_{i+N,j}=w_1\psi_{i,j},\quad
\psi_{i,j+M}=w_2\psi_{i,j}.
\end{align}
Let $\mathcal{V}_{w_1,w_2}$ be the $NM$-dimensional vector space of functions $f: \mathbb{Z}^2 \to \mathbb{C}$ satisfying the quasi-periodic boundary conditions:
$$f_{i+N, j} = w_1 f_{i,j}, \quad f_{i, j+M} = w_2 f_{i,j}, \quad \forall i,j \in \mathbb{Z}.$$
Since any function in $\mathcal{V}_{w_1,w_2}$ is uniquely determined by its values on the fundamental domain $\{1,\ldots,N\} \times \{1,\ldots,M\}$, the action of the partial difference operator $\mathcal{D}$ naturally restricts to $\mathcal{V}_{w_1,w_2}$. This restriction can be represented by an $NM \times NM$ matrix, known as the Bloch matrix $\mathcal{D}(w_1,w_2)$. The associated spectral curve $\Gamma_{\mathcal{D}}$ is then defined by the algebraic equation:
\begin{align*}
\Gamma_\D=\{(w_1,w_2)\in(\mathbb{C}^*)^2\mid\det \D(w_1,w_2)=0\}.
\end{align*}
For a generic operator $\D$,
$\Gamma_\D$ is a Riemann surface of genus $(N-1)(M-1)$ \cite{Krichever1985}.
\begin{lemma}
The kernel of $\D\in\NDDO(\J)$ is an infinite-dimensional  linear space.
\end{lemma}
\begin{proof}
	By definition, the spectral curve $\Gamma_{\mathcal{D}}$ is a complex algebraic curve, which necessarily contains infinitely many distinct points. For each point $(w_1,w_2) \in \Gamma_{\mathcal{D}}$, there exists at least one non-trivial Bloch solution $\psi{(w_1,w_2)} \in \Ker \mathcal{D}$ satisfying the Floquet-Bloch conditions $T_1^N \psi{(w_1,w_2)} = w_1 \psi{(w_1,w_2)}$ and $T_2^M \psi{(w_1,w_2)} = w_2 \psi{(w_1,w_2)}$.
	
	These Bloch solutions are joint eigenfunctions of the commuting translation operators $T_1^N$ and $T_2^M$. Since eigenvectors corresponding to distinct joint eigenvalues $(w_1, w_2)$ are strictly linearly independent, the infinite family of solutions $\{\psi{(w_1,w_2)}\}_{(w_1,w_2)\in\Gamma_{\mathcal{D}}}$ constitutes an infinite linearly independent set within $\Ker \mathcal{D}$. Therefore, $\Ker \mathcal{D}$ must be an infinite-dimensional vector space.
\end{proof}
\begin{lemma}\label{lemma_wz}
The action \eqref{eq_lr} induces a rescaling of the multipliers in \eqref{eq_omega}
\begin{align}\label{eq_act_w}
(w_1,w_2)\mapsto(\frac{w_1}{z_1},\frac{w_2}{z_2}),
\end{align}
where $(z_1,z_2)$ is defined in \eqref{eq_h}.
Since the map is globally analytic,
the algebraic structure of the curve $\Gamma_\D$ remains invariant up to this scaling.
\end{lemma}
\begin{proof}
Suppose that $(\alpha,\beta)\in H\tilde{\times}H$ with  $(z_1,z_2)$ as in \eqref{eq_h}.
Let $\psi_{i,j}$ be a Bloch solution of $\D$ with multipliers $(w_1,w_2)$.
Then  $\Psi_{i,j}:=\alpha^{-1}_{i,j}\psi_{i,j}$ is a solution for the transformed operator $\beta^{-1}\D\alpha$.
The Lemma follows from the relations:
\begin{align*}
T_1^N\Psi_{i,j}=T_1^N\alpha^{-1}_{i,j}\psi_{i,j}=\frac{w_1}{z_1}\Psi_{i,j},\\
T_2^M\Psi_{i,j}=T_2^M\alpha^{-1}_{i,j}\psi_{i,j}=\frac{w_2}{z_2}\Psi_{i,j}.
\end{align*} 
\end{proof}
In the following,
we  consider the case $\J=\{(0,0),\,(0,1),\,(1,0),\,(1,1)\}$ and $d=3$ as in Example \ref{exam_corrugated} for simplicity.
The other cases can be treated similarly.
\begin{lemma}\label{lem:zariski_dense}
	The subset of twisted $\J$-corrugated $(N,M)$-hedra whose associated $4 \times 4$ monodromy matrices $\phi_1$ and $\phi_2$ possess distinct eigenvalues (and are therefore diagonalizable) is Zariski dense in the space $\mathcal{P}_{N,M}^3(\J)$ of twisted $\J$-corrugated $(N,M)$-hedra in $\mathbb{RP}^3$ modulo projective transformations.
\end{lemma}
\begin{proof}
	According to Definition \ref{def_twi}, a twisted $\J$-corrugated $(N,M)$-hedron is determined by its vertex configurations along with its associated monodromy matrices $\phi_1, \phi_2 \in \mathrm{PGL}(4)$. 
	Therefore, the configuration space $\mathcal{P}_{N,M}^3(\J)$ is parametrized by the homogeneous coordinates of the vertices in a fundamental domain together with the matrix entries of $\phi_1$ and $\phi_2$. These parameters are subject to a set of algebraic constraints, namely the $\J$-corrugated linear dependence conditions and the quasi-periodic boundary matching equations. 	
	By imposing the strict non-degeneracy condition (where no three vertices in any elementary quadrilateral are collinear), this configuration space is identified with a principal open subset of an affine algebraic variety, rendering it irreducible.

	Since the entries of $\phi_1$ and $\phi_2$ are  coordinate functions of this affine parameter space, the discriminants of their characteristic polynomials, denoted by $\Delta_1$ and $\Delta_2$, are globally defined polynomial functions on this variety. The condition that the monodromy matrices possess multiple eigenvalues precisely corresponds to the algebraic equation $\Delta_1 \Delta_2 = 0$, which defines a Zariski closed subset.

	To show this closed set is proper, it suffices to construct a single configuration where $\Delta_1 \Delta_2 \neq 0$. Consider the constant-coefficient operator $\mathcal{D}_0 = 1 + 2T_2 + 3T_1 + 4T_1T_2 \in \NDDO(\J)$. Its plane wave solutions $\psi_{i,j} = \lambda_1^i \lambda_2^j$ satisfy the rational dispersion relation $1 + 2\lambda_2 + 3\lambda_1 + 4\lambda_1\lambda_2 = 0$, with the corresponding $(N,M)$-periodic joint Floquet multipliers given by $w_1 = \lambda_1^N$ and $w_2 = \lambda_2^M$.
	
	By generically selecting four pairs $(\lambda_1^{(k)}, \lambda_2^{(k)})_{k=1}^4$ on this rational curve, we can ensure their $N$-th and $M$-th powers $\{w_{1,k}\}$ and $\{w_{2,k}\}$ are mutually distinct. Assembling the four associated Bloch solutions yields a valid $\J$-corrugated $(N,M)$-hedron whose monodromy matrices strictly have distinct eigenvalues. Therefore, $\Delta_1 \Delta_2 \not\equiv 0$, making the diagonalizable configurations a Zariski dense open subset.
\end{proof}

\begin{proposition}\label{pro_one-to-one}
	Let $\J=\{(0,0),\,(0,1),\,(1,0),\,(1,1)\}$.
There is a one-to-one correspondence between the following spaces:

(1) The space of generic twisted $\J$-corrugated $(N, M)$-hedra in $\mathbb{RP}^3$ modulo projective transformations. 
Here, ``generic'' means the polyhedra are strictly non-degenerate, and the associated monodromy matrices are diagonalizable with mutually distinct eigenvalues.

(2) The space of pairs $(\mathcal{D}, \sum_{k=1}^4 Q_k)$, where $\mathcal{D} \in \NDDO(\J)$ is a non-degenerate $(N, M)$-periodic partial difference operator whose associated spectral curve $\Gamma_{\mathcal{D}}$ is generically smooth, and $\sum_{k=1}^4 Q_k$ is a real divisor of four distinct marked points on the spectral curve $\Gamma_{\mathcal{D}}$, invariant under the complex conjugation $\tau$ (i.e., $\tau(\sum_{k=1}^4 Q_k) = \sum_{k=1}^4 Q_k$).
This space is taken modulo the gauge action:
$$
(\mathcal{D}, \sum_{k=1}^4 Q_k) \mapsto (\beta^{-1}\mathcal{D}\alpha, \sum_{k=1}^4 Q_k')
$$
where $(\alpha,\beta) \in H\tilde{ \times }H$ has the monodromy $(z_1,z_2)$ as in \eqref{eq_h}, 
and each $Q'_k$ is the transformed point according to \eqref{eq_act_w}.
\end{proposition}

\begin{proof}
Given a projective equivalence class of general twisted $\J$-corrugated $(N,M)$-hedron with monodromies $\phi_1$ and $\phi_2\in\mathrm{GL}(4,\mathbb{R})$.
Let $\{v_{i,j}\in\mathbb{RP}^3\}\in\mathcal{P}_{N,M}^3(\J)$ be an arbitrary representative  of this class.
Lift each point $v_{i,j}\in\mathbb{RP}^3$ to a  vector $V_{i,j}\in\mathbb{R}^4$.
From the $\J$-corrugated condition, it follows that for any $(i,j)\in\mathbb{Z}^2$, the set of vectors $\{V_{i+n,j+m}\mid(n,m)\in\J\}$ is linearly dependent.
Specially,
there exist coefficients such that:
\begin{align}\label{eq_d}
a_{i,j}V_{i,j}+b_{i,j}V_{i,j+1}+c_{i,j}V_{i+1,j}+d_{i,j}V_{i+1,j+1}=0.
\end{align}
This is equivalent to $\D V=0$,
where $V$ is the bi-infinite matrix of  $V_{i,j}\in\mathbb{R}^4$,
and the operator $\D$ is defined as:
\begin{align*}
\D=a+bT_2+cT_1+dT_1T_2.
\end{align*}
Since the matrix $\{V_{i,j}\}$ is $(N,M)$-quasi-periodic,
the resulting operator $\D$ is $(N,M)$-periodic.
Crucially, the non-degenerate geometric condition---no $|\J|-1$ vertices from the set $\{v_{i+n,j+m}\}_{(n,m)\in \J}$ belong to a $(|\J|-3)$-dimensional projective space---is algebraically equivalent to stating that no $|\J|-1$ vectors from the set $\{V_{i+n,j+m}\}_{(n,m)\in \J}$ are linearly dependent. Consequently, in the linear relation $\mathcal{D}V = 0$, the coefficients $a_{i,j}$, $b_{i,j}$, $c_{i,j}$, $d_{i,j}$ must be strictly non-zero for all $i,j\in\mathbb{Z}$. 
This exactly dictates that $\mathcal{D}$ belongs to the space $\NDDO(\J)$.

The four components of $\{V_{i,j}\in\mathbb{R}^4\}$ are four solutions in $\Ker\D$.
Since $T_1^N(V_{i,j})=\phi_1 V_{i,j}$ and $T_2^M(V_{i,j})=\phi_2 V_{i,j}$,
these four solutions span a four-dimensional subspace of $\Ker\D$ invariant under $T_1^N$ and $T_2^M$.
By our genericity assumption (which constitutes a Zariski dense open subset as proven in Lemma \ref{lem:zariski_dense}), the monodromy matrices $\phi_1$ and $\phi_2$ are simultaneously diagonalizable with mutually distinct joint eigenvalues. Therefore, this $4$-dimensional invariant subspace canonically decomposes into four $1$-dimensional eigenspaces. These eigenspaces uniquely identify four linearly independent standard Bloch solutions, which exactly correspond to four distinct points $Q_1, Q_2, Q_3, Q_4$ on the smooth spectral curve $\Gamma_{\mathcal{D}}$.

Furthermore,
we can rescale each  $V_{i,j}$ by a scalar and  multiply each  equation in \eqref{eq_d} by a scalar factor.
Thus, $\D$ is defined up to the left-right action \eqref{eq_lr}.
According to Lemma \ref{lemma_wz},
the divisor \(\sum_{k=1}^4 Q_k\) is determined up to the corresponding rescaling \eqref{eq_act_w}.

Conversely, given a pair $(\mathcal{D}, \sum_{k=1}^4 Q_k)$, a standard result in Krichever's algebra-geometric theory of difference operators indicates that the kernel of the Bloch matrix at any regular point on the curve is exactly one-dimensional. Therefore, each point $Q_k$ uniquely identifies a Bloch solution $\psi(Q_k)$ up to a non-zero scalar multiplier. 
Let $V_{i,j} = [\psi_{i,j}(Q_1) : \psi_{i,j}(Q_2) : \psi_{i,j}(Q_3) : \psi_{i,j}(Q_4)]^T \in \mathbb{CP}^3$ be the projective coordinate vector. 
It satisfies
\begin{align*}
T_1^NV_{i,j}=\phi_1V_{i,j},\quad
T_2^MV_{i,j}=\phi_2V_{i,j},
\end{align*}
where $Q_k=(Q_{k,1},Q_{k,2})$ for $k=1,2,3,4$, and $\phi_1, \phi_2 \in \mathrm{GL}(4,\mathbb{C})$ are the diagonal matrices $\operatorname{diag}(Q_{1,1}, Q_{2,1}, Q_{3,1}, Q_{4,1})$ and $\operatorname{diag}(Q_{1,2}, Q_{2,2}, Q_{3,2}, Q_{4,2})$  respectively.
Since $V_{i,j}$ generally has complex entries, we need to apply a  suitable transformation to restrict the geometry to $\mathbb{RP}^3$.

For a real-valued operator $\mathcal{D}$, 
the BA function satisfies $\overline{\psi_{i,j}(P)} = \psi_{i,j}(\tau(P))$,
where $\tau$ is the complex conjugation. 
Since the divisor $\sum_{k=1}^4 Q_k$ is real, 
the entries of $V_{i,j}$ consist of complex conjugate pairs.
To ensure the geometric realization lies in $\mathbb{RP}^3$, we perform a canonical change of the projective frame. Specifically, we apply a constant  matrix $S \in \mathrm{GL}(4, \mathbb{C})$ that maps these conjugate pairs to their real and imaginary parts, 
\begin{align*}
	\tilde{V}_{i,j} := S
	\begin{pmatrix}
		\psi_{i,j}(Q_1)\\
		\psi_{i,j}(Q_2)\\
		\psi_{i,j}(Q_3)\\
		\psi_{i,j}(Q_4)
	\end{pmatrix},
\end{align*}
yielding a purely real vector $\tilde{V}_{i,j} = S V_{i,j} \in \mathbb{R}^4$.
Crucially, it corresponds to a similarity transformation of the associated monodromy matrices,
that is,
\begin{align*}
T_1^N\tilde{V}_{i,j}=\tilde{\phi}_1\tilde{V}_{i,j},\quad
T_2^M\tilde{V}_{i,j}=\tilde{\phi}_2\tilde{V}_{i,j},
\end{align*}
where $\tilde{\phi}_i=S\phi_i S^{-1}\in\mathrm{GL}(4,\mathbb{R})$, $i=1,\,2$.
Therefore, the positions of the marked points $Q_k$ remain strictly invariant under the transformation $S$. 
By construction, $\mathcal{D}\tilde{V} = 0$, ensuring the $\J$-corrugated property. 
Consequently, $\{\tilde{V}_{i,j}\}_{i,j\in\mathbb{Z}}$ is the twisted $\J$-corrugated $(N,M)$-hedron we needed.

If we select a different representative $\tilde{\mathcal{D}} = \beta^{-1}\mathcal{D}\alpha$ from  the same gauge class, the new Bloch solutions satisfy ${\Psi}(Q'_k) = \alpha^{-1} \psi(Q_k)$, where $Q'_k$ is the rescaled multiplier point \eqref{eq_act_w}. 
In $\mathbb{RP}^3$, $\{\psi(Q_k)\}_{k=1}^4$ and $\{\Phi(Q_k')\}_{k=1}^4$ correspond to the same $(N,M)$-hedron,
since $\alpha_{i,j}^{-1}$ is a scalar factor at each grid point. 
The action $(w_1, w_2) \to (w_1/z_1, w_2/z_2)$ ensures that the configuration of divisors $\sum_{k=1}^4 Q_k$ is consistently tracked across the gauge orbit. 
Thus the projective equivalence class of the resulting $(N,M)$-hedron is uniquely determined by the pair $(\mathcal{D}, \sum_{k=1}^4 Q_k)$, completing the bijection.
\end{proof}

\section{Poisson-Lie group of pseudo partial difference operators}\label{sec_poi}

We define an $(N,M)$-periodic pseudo partial difference operator as a formal Laurent series in terms of the left shift operators $T_1$, $T_2$,
whose coefficients are $(N,M)$-periodic bi-infinite matrices.
Explicitly,
every such operator is of the form
\begin{align*}
	\sum_{(n,m)\succ (n_0,m_0)}a^{n,m}T_1^nT_2^m,
\end{align*}
where $(n,m)\in\mathbb{Z}^2$,
$T_1$, $T_2$ are the left shift operators on $\mat$,
while each $a^{m,n}$ is an $(N,M)$-periodic bi-infinite matrix of real numbers.

We denote the set of $(N,M)$-periodic pseudo partial difference operators by $\Psi\DO$.
It is an associative algebra with respect to addition and multiplication of operators.
The operator is invertible if the coefficient $a^{n_0,m_0}$ of lowest power is a bi-infinite matrix none of whose elements vanishes.
The set of invertible $(N,M)$-periodic pseudo partial difference operators is denoted by $\mathrm{I}\Psi\DO$,
which is a group with respect to multiplication.
It can be regarded as an infinite-dimensional Lie group with Lie algebra $\Psi\DO$.

\subsection{General Poisson-Lie groups}
This subsection is a brief introduction to the theory of Poisson-Lie group.
They are taken from \cite{refactorization2020,reyman1994group} and occasionally modified to fit our needs.

A Poisson-Lie group is a Lie group equipped with a compatible Poisson structure.
It is the key object for group-theoretic constructions of lots of integrable systems.
\begin{definition}
	A {Poisson-Lie group} is a Lie group $G$ endowed with a Poisson structure $\pi$, such that the group multiplication map
	$$
	m: G \times G \to G, \quad m(g,h) = gh
	$$
	is a Poisson map, where $G \times G$ is equipped with the product Poisson structure.
\end{definition}
Assume that $G$ is a Poisson-Lie group and $\mathcal{B}$ is its Lie algebra.
Then $G$ is coboundary, if its Poisson tensor at every point $g\in G$ can be written as
\begin{align}\label{def_pi}
	\pi_g = \frac{1}{2}\left( (\lambda_g)_* \hat{r} - (\rho_g)_* \hat{r} \right) 
\end{align}
for a fixed element $\hat{r} \in \mathcal{B} \wedge \mathcal{B}$ called the {classical r-matrix}. 
Here $\lambda_g$ and $\rho_g$ are, respectively, the left and right translation by $g$.

A simple sufficient condition for the Poisson tensor \eqref{def_pi} to satisfy the Jacobi identity is that the r-matrix  satisfies the {modified Yang-Baxter equation (mYBE)}:
\begin{align*}
	[rx,ry]-r[rx,y]-r[x,ry]=-[x,y],\quad
	\forall\,x,\,y\in\mathcal{B}.
\end{align*}
A particularly important class of  r-matrices comes from a decomposition of the Lie algebra:
\begin{proposition}\label{pro_r}
	Let $\mathcal{B}$ be a Lie algebra endowed with an invariant inner product.
	Assume that $\mathcal{B}$, as a vector space,
	can be written as a direct sum of three subalgebras $\mathcal{B}_{>0}$,
	$\mathcal{B}_0$ and $\mathcal{B}_{<0}$,
	such that $[\mathcal{B}_0,\mathcal{B}_{>0}]\subset\mathcal{B}_{>0}$,
	$[\mathcal{B}_0,\mathcal{B}_{<0}]\subset\mathcal{B}_{<0}$, the subalgebras $\mathcal{B}_{>0}$, $\mathcal{B}_{<0}$ are isotropic,
	and $\mathcal{B}_0$ is orthogonal to both $\mathcal{B}_{>0}$ and $\mathcal{B}_{<0}$.
	Then $r:=p_{>0}-p_{<0}$ satisfies the modified Yang-Baxter equation,
	where $p_{>0}$, $p_{<0}$ are projectors $\mathcal{B}\rightarrow\mathcal{B}_{>0}$, $\mathcal{B}\rightarrow\mathcal{B}_{<0}$ respectively.
	Then the group $G$ of the Lie algebra $\mathcal{B}$ is a factorizable Poisson-Lie group.
\end{proposition}
For a Lie group $G$ of invertible elements in an associative algebra $A$ with an  invariant inner product $\langle\cdot,\cdot\rangle$, the coboundary Poisson-Lie tensor takes the explicit computable form:
\begin{align}\label{eq_pi}
	\pi_{g}(x,y)=\frac{1}{2}\left(
	\langle r(xg),yg\rangle
	-
	\langle r(gx),gy\rangle
	\right),\quad
	\forall\,g\in G,\,x,\,y\in A.
\end{align}
where $r:A\to A$ is the skew-symmetric linear operator corresponding to the classical $r$-matrix
and we identify the cotangent space $T_g^*G$ with the tangent space $T_gG=A$.
Here the right hand side of \eqref{eq_pi} is defined for every $g\in A$.
Therefore it is a Poisson bracket on the whole associative algebra $A$.

Finally,
we give two propositions on the coboundary Poisson-Lie groups.
\begin{proposition}\label{pro_auto}
Let $\sigma:\,G\rightarrow G$ be an automorphism of a coboundary Poisson-Lie group.
If the differential of \(\sigma\) at the identity preserves the r-matrix $\hat{r}$,
then \(\sigma\) is a Poisson map.
\end{proposition}

\begin{proposition}\label{pro_ad}
Let $G$ be a Lie group endowed with a coboundary Poisson structure $\pi$ defined by an $r$-matrix $\hat{r}$.
Then the Poisson structure $\pi$ vanishes at $g\in G$ iff $(\mathrm{Ad}_g)_*\hat{r}=\hat{r}$.
\end{proposition}
\begin{proof}
Since $(\mathrm{Ad}_g)_*\hat{r}=(\rho_g^{-1})_*(\lambda_g)_*\hat{r}$,
we have $(\mathrm{Ad}_g)_*\hat{r}=\hat{r}$ iff $(\lambda_g)_*\hat{r}=(\rho_g)_*\hat{r}$, i.e. $\pi_g=0$.
\end{proof}
\subsection{The Poisson-Lie group of pseudo partial difference operators}

The Lie algebra of the group $\mathrm{I}\Psi\DO$ is the space $\Psi\DO$ of all $(N,M)$-periodic pseudo partial difference operators.
The algebra has an invariant inner product defined by
\begin{align}\label{def_inn}
	\langle \D_1,\D_2\rangle=\Tr\, \D_1\D_2,\quad
	\forall\, \D_1,\,\D_2\in\Psi\DO,
\end{align}
where the trace of an $(N,M)$-periodic pseudo partial difference operator is defined as
\begin{align*}
	\Tr\left(\sum_{(n,m)\succ(n_0,m_0)}a^{n,m}T_1^nT_2^m\right):=
	\sum_{i=1}^N\sum_{j=1}^Ma^{0,0}_{i,j}.
\end{align*}
The inner product \eqref{def_inn} is non-degenerate and invariant in the associative algebra sense,
i.e. $\langle\D_1,\D_2\D_3\rangle=\langle\D_1\D_2,\D_3\rangle$.
Furthermore,
it is symmetric, that is, $\Tr\,\D_1\D_2=\Tr\,\D_2\D_1$.

Represent the algebra $\Psi\DO$ of $(N,M)$-periodic partial difference operators as the sum of three subalgebras:
\begin{align*}
	\mathcal{B}_{\prec(0,0)}:=\Psi\DO(\mathbb{Z}^2_{\prec(0,0)}),\quad
	\mathcal{B}_{0,0}:=\Psi\DO({\{(0,0)\}}),\quad
	\mathcal{B}_{\succ(0,0)}:=\Psi\DO(\mathbb{Z}^2_{\succ(0,0)}),
\end{align*}
where $\mathbb{Z}^2_{\succ(0,0)}$ stands for elements in $\mathbb{Z}^2$ bigger than $(0,0)$ and $\mathbb{Z}^2_{\prec(0,0)}$ for elements smaller than $(0,0)$.
The decomposition satisfies the requirements of Proposition \ref{pro_r}, so we get an $r$-matrix $r:=p_{>0}-p_{<0}$ and hence a Poisson-Lie structure on $\mathrm{I}\Psi\DO$.

Let $\nabla^l f$, $\nabla^rf$ be the left and right gradients for a function $f\in C^\infty(\mathrm{I}\Psi\DO)$: 
\begin{align}\label{def_lr}
	\langle \nabla^lf(x),y\rangle=\left.\frac{\d}{\d t}\right|_{t=0}f(\exp(ty)x),\quad
	\langle \nabla^rf(x),y\rangle=\left.\frac{\d}{\d t}\right|_{t=0}f(x\exp(ty)).
\end{align}
Then the bracket given by 
\begin{align}\label{def_sklyanin}
	\{f_1,f_2\}
	=\frac{1}{2}
	\left(
	\langle r(\nabla^rf_1),\nabla^rf_2\rangle
	-\langle r(\nabla^lf_1),\nabla^lf_2
	\rangle
	\right),
\end{align}
defines a Poisson-Lie bracket, which is called the Sklyanin bracket.
\begin{proposition}\label{pro_sub}
The group $\mathrm{I}\Psi\DO$ of $(N,M)$-periodic invertible pseudo partial difference operators, together with the Poisson structure defined by \eqref{def_sklyanin}, is a Poisson-Lie group.
This structure has the following properties:
\begin{enumerate}
	\item Assume that $\J\subset\mathbb{Z}^2$ is a finite subset that consists of consecutive elements (in the lexicographical order).
	Then the subset $\mathrm{I}\DO(\J):=\mathrm{I}\Psi\DO\cap\DO(\J)$ is a Poisson submanifold of $\mathrm{I}\Psi\DO$.
	\item 
	If $\J$ is a one-point set,
	then the restriction of $\pi$ to $\mathrm{I}\Psi\DO(\J)$ is zero.
	\item 
	The Poisson structure is invariant under the left-right action \eqref{eq_lr} of the group $H\tilde{\times}H$ of pairs of non-vanishing $(N,M)$-quasi-periodic bi-infinite matrices with the same monodromy.
\end{enumerate}
\end{proposition}
\begin{proof}
	Recall that the formula \eqref{eq_pi} can be rewritten as
	\begin{align*}
		\pi_{g}(x,y)&=\frac{1}{2}\left(
		\langle r(xg),yg\rangle
		-
		\langle r(gx),gy\rangle
		\right)\\
		&=\frac{1}{2}\left(
		\langle gr(xg),y\rangle
		-
		\langle r(gx)g,y\rangle
		\right)
		=\frac{1}{2}\langle gr(xg)-r(gx)g,y\rangle.
	\end{align*}
	Thus,
	we can treat the tensor $\pi$ as a map $A\rightarrow A$:
	\begin{align*}
		\pi_g(x)=\frac{1}{2}\left(g r(xg)-r(gx)g\right),
	\end{align*}
	so we have
	$\pi_g(x,y)=\langle \pi_g(x),y\rangle.$
	In the case of $A=\Psi\DO$,
	the Poisson tensor $\pi_\D$ is given by
	\begin{align}\label{eq_dd}
		\pi_\D(Q)=\frac{1}{2}\left(\D r(Q\D)-r(\D Q)\D\right),
	\end{align}
	where $\D\in\mathrm{I}\Psi\DO$ and $Q\in\Psi\DO$.
	It is sufficient to prove that for $\D\in\mathrm{I}\Psi\DO(\J)$ the image of the Poisson tensor \eqref{eq_dd} belongs to  the tangent space to $\mathrm{I}\Psi\DO(\J)$ at $\D$.
	Equivalently,
	we need to prove that \eqref{eq_dd} is supported in $\J$ whenever $\D$ is supported in $\J$.
	Note that \eqref{eq_dd} is invariant under the transformation $r\rightarrow r\pm Id$.
	Since $r+Id=2p_{>0}+p_0$,
	the operator $r(Q\D)$ and $r(\D Q)$ have terms of  $T_1^nT_2^m$ with $(n,m)\succeq(0,0)$.
	Thus we get 
	\begin{align*}
		\min \,\text{supp}\,\pi_\D(Q)\succeq\min\,\text{supp}\,\D=\min\,\J.
	\end{align*}
	Similarly,
	consider the case of $r-Id$,
	we get $\text{supp}\pi_\D(Q)\preceq\max\,\J$.
	In a word,
	we have $\text{supp}\,\pi_\D(Q)\subset[\min\,\J,\max\,\J]=\J$ as desired.
	
	To prove the second statement,
	notice that if $\D$ is supported in a one-point set,
	then  conjugation by $\D$ preserves the subalgebras $\mathcal{B}_\pm$ and $\mathcal{B}_0$, as well as the inner product on $\Psi\DO$.
	According to Proposition \ref{pro_ad},
	we have $\pi(\D)=0$.
	
	To prove the third statement,
	we represent the left-right action \eqref{def_lr} as a superposition of two actions: one is of the same form, but with $(N,M)$-periodic $\alpha$, $\beta$,
	while the other one is conjugation action by a $(N,M)$-quasi-periodic $\gamma$.
	The former one is Poisson since the Poisson structure vanishes on scalar bi-infinite matrices,
	while the later is also Poisson because conjugation by bi-infinite matrices preserves  the subalgebras $\mathcal{B}_\pm$ and $\mathcal{B}_0$, as well as the inner product on $\Psi\DO$, and hence Poisson by Proposition \ref{pro_auto}.
	Thus the left-right action \eqref{def_lr} is Poisson.
\end{proof}
\begin{remark}
The central functions are defined as $f_{i,j,k}(\D):=\Tr\, T_{1}^{iN}T_2^{jM}\D^k$,
where $i,j\in\mathbb{Z}$ and $k\in\mathbb{N}$.
They  Poisson commute.
\end{remark}
\subsection{Explicit formulas for Poisson brackets of coordinate functions}
In this subsection, we utilize the Poisson-Lie framework established above to derive explicit coordinate-wise formulas for the Poisson brackets of \((N,M)\)-periodic pseudo partial difference operators.

Before proceeding with explicit coordinate computations, we clarify a slight abuse of notation regarding subscripts, which is standard but potentially ambiguous. When a coefficient $a$ appears without subscripts in an operator expression (e.g., $aT_1$), it represents the entire bi-infinite  matrix. 
If subscripts are attached to a coefficient acting as a matrix within an operator expression (e.g., $a_{n,m}T_1$), they denote a global grid shift of the matrix, effectively meaning the matrix elements are evaluated at shifted indices, $(a_{n,m})_{i,j} = a_{i+n,j+m}$. Conversely, when a coefficient $a$ appears in isolation as a scalar variable (such as the coordinate definition $a_{i,j}$), the subscripts denote its specific scalar component at the lattice point $(i,j)$.

\begin{proposition}
	Let $x=\sum_{i,j}a^{i,j}T_1^iT_2^j$,
	$y=\sum_{n,m}y^{n,m}T_1^nT_2^m$ be two generic elements in $\Psi\DO$.
	Let $f \in C^\infty(\Psi \DO)$ be a smooth function,
	then the left gradient is $\nabla^lf(x)=\sum_{n,m}A^{n,m}T_1^nT_2^m$,
	and the right gradient is $\nabla^rf(x)=\sum_{n,m}B^{n,m}T_1^nT_2^m$,
	where
	\begin{align*}
		A_{r,s}^{n,m}:=
		\sum_{i,j}\frac{\partial f}{\partial a_{r+n,s+m}^{i-n,j-m}}a_{r,s}^{i,j},\quad
		B_{r,s}^{n,m}:=\sum_{i,j}
		\frac{\partial f}{\partial a^{i-n,j-m}_{r+n-i,s+m-j}}a^{i,j}_{r+n-i,s+m-j}.
	\end{align*}
\end{proposition}
\begin{proof}
	The left hand side of the first equation in \eqref{def_lr} is
	\begin{align}\label{eq1}
		\langle \nabla^lf(x),y\rangle=
		\Tr\left(
		\sum_{n,m,i,j}y^{n,m}A_{n,m}^{i,j}T_1^{i+n}T_2^{j+m}
		\right)
		=
		\sum_{n,m,r,s}y_{r,s}^{n,m}A^{-n,-m}_{r+n,s+m}.
	\end{align}
	The right hand side of the first equation in \eqref{def_lr} is
	\begin{align}\label{eq2}
		\left.\frac{\d}{\d t}\right|_{t=0}f(\exp(ty)x)
		=
		\sum_{i,j,n,m,r,s}\frac{\partial f}{\partial a_{r,s}^{i+n,j+m}}a^{i,j}_{r+n,s+m}y_{r,s}^{n,m}.
	\end{align}
	Comparing \eqref{eq1} and \eqref{eq2},
	we get
	\begin{align*}
		A_{r+n,s+m}^{-n,-m}=
		\sum_{i,j}\frac{\partial f}{\partial a_{r,s}^{i+n,j+m}}a^{i,j}_{r+n,s+m},
	\end{align*}
	or equivalently,
	\begin{align*}
		A_{r,s}^{n,m}=
		\sum_{i,j}\frac{\partial f}{\partial a_{r+n,s+m}^{i-n,j-m}}a_{r,s}^{i,j}.
	\end{align*}
	The right gradient can be obtained similarly.
\end{proof}
Having established the expressions for the left and right gradients of a general smooth function, 
we now apply this result specifically to the fundamental coordinate functions $a_{p,q}^{k,l}$ to obtain the explicit gradients required for computing the Poisson structure.
\begin{coro}\label{coro_coordinate}
	Let $x=\sum_{i,j}a^{i,j}T_1^iT_2^j$, 
	the left gradient of the coordinate function $a_{p,q}^{k,l}$
	is 
	\begin{align*}
		\nabla^la_{p,q}^{k,l}
		=
		\sum_{n,m}a^{n+k,m+l}e_{p-n,q-m}T_1^nT_2^m,
	\end{align*}
	while the right gradient is
	\begin{align*}
		\nabla^ra_{p,q}^{k,l}
		=
		\sum_{n,m}a^{n+k,m+l}_{-k,-l}e_{p+k,q+l}T_1^nT_2^m,
	\end{align*}
	where $e_{i,j}\in\mat$ denotes the standard basis matrix with $1$ at the $(i,j)$-th entry and $0$ elsewhere.
\end{coro}
Now  we can substitute these gradients directly into the underlying Sklyanin bracket defined in \eqref{def_sklyanin} to explicitly evaluate the Poisson brackets between any pair of coordinate functions.
\begin{theorem}\label{the_poi}
	Let $x=\sum_{i,j}a^{i,j}T_1^iT_2^j$,
	the Poisson bracket for coordinate functions is
	\begin{align}\label{Poisson_coordinate}
		\begin{split}
			\{a_{p,q}^{k,l},a_{r,s}^{u,v}\}
			=&
			\frac{1}{2}a_{p,q}^{u+r-p,v+s-q}a_{r,s}^{k+p-r,l+q-s}\\
			&\cdot(
			\sgn_{r+u-p-k}+\delta_{p+k,r+u}\sgn_{s+v-q-l}
			-\sgn_{p-r}-\delta_{p,r}\sgn_{q-s}),
		\end{split}
	\end{align}
	where $\delta_{x,y}$ denote the Kronecker delta, defined to be $1$ if $x=y$
	and $0$
	otherwise,
	while $\sgn_x=x/|x|$ for $x\neq 0$ and $\sgn_0=0$. 
\end{theorem}
\begin{proof}
	According to \eqref{def_sklyanin},
	we have
	\begin{align}\label{pd0}
		\{a_{p,q}^{k,l},a_{r,s}^{u,v}\}
		=
		\frac{1}{2}
		\left(
		\langle r(\nabla^ra_{p,q}^{k,l}),\nabla^ra_{r,s}^{u,v}\rangle
		-\langle r(\nabla^la_{p,q}^{k,l}),\nabla^la_{r,s}^{u,v}
		\rangle
		\right).
	\end{align}
	Using the formulas in  Corollary \ref{coro_coordinate},
	we have
	\begin{align}\label{pd1}
		\begin{split}
			&\langle r(\nabla^ra_{p,q}^{k,l}),\nabla^ra_{r,s}^{u,v}\rangle\\
			=&\,
			\langle
			r(
			\sum_{n,m} a_{-k,-l}^{n+k,m+l}e_{p+k,q+l}T_1^nT_2^m),
			\,
			\sum_{n,m}
			a_{-u,-v}^{n+u,m+v}e_{r+u,s+v}T_1^nT_2^m
			\rangle
			\\
			=&
			\left(\sum_{n>0,m}-\sum_{n<0,m} \right)
			a_{-k,-l}^{n+k,m+l}e_{p+k,q+l}a_{n-u,m-v}^{u-n,v-m}e_{r+u-n,s+v-m}\\
			&+\left(\sum_{m>0}-\sum_{m<0}\right)
			a_{-k,-l}^{k,m+l}e_{p+k,q+l}a_{-u,m-v}^{u,v-m}e_{r+u,s+v-m}\\
			=&\,
			a_{p,q}^{u+r-p,v+s-q}a_{r,s}^{k+p-r,l+q-s}\left(\sgn_{r+u-p-k}+\delta_{p+k,r+u}\sgn_{s+v-q-l}\right).
		\end{split}
	\end{align}
	Similarly,
	we get
	\begin{align}\label{pd2}
		\langle r(\nabla^la_{p,q}^{k,l}),\nabla^la_{r,s}^{u,v}
		\rangle
		=a_{p,q}^{u+r-p,v+s-q}a_{r,s}^{k+p-r,l+q-s}
		(
		\sgn_{p-r}+\delta_{p,r}\sgn_{q-s}).
	\end{align}
	Substituting \eqref{pd1} and \eqref{pd2} into \eqref{pd0},
	we obtain the Poisson bracket.
\end{proof}

It is worth emphasizing that while the pseudo-difference operators and their gradients $\nabla^{l}f, \nabla^{r}f$ (as derived in Corollary \ref{coro_coordinate}) are formal Laurent series involving infinite summations, the explicit Poisson brackets of coordinate functions in Theorem \ref{the_poi} evaluate to   finite expressions. 
This remarkable algebraic collapse occurs because the trace inner product perfectly pairs specific matching terms from the infinite series, which, together with the projection operator $r = p_{>0} - p_{<0}$, strictly annihilates the rest of the formal sum, yielding the clean finite formulas in \eqref{Poisson_coordinate}.
\subsection{Reductions of partial difference operators}
In this subsection, we investigate some reductions for $(N,M)$-periodic partial difference operators: the reduction to ordinary (1D) difference operators and the reduction to sparse operators supported on sub-lattices.

A general pseudo partial difference operator is $x=\sum_{i,j}a^{i,j}T_1^iT_2^j$.
Here,
we consider a special case $i\equiv 0$ and let $T=T_2$,
that is $x=\sum a^iT^i$.
\begin{proposition}
	Let $x=\sum a^iT^i$, then
	\begin{align*}
		\nabla^lf(x)=\sum A^iT^i,\quad
		A_n^i:=\sum_m \frac{\partial f}{\partial a_{n+i}^m}a_n^{m+i},
	\end{align*}
	and
	\begin{align*}
		\nabla^rf(x)=\sum B^iT^i,\quad
		B^i_n:=\sum_m \frac{\partial f}{\partial a_{n-m}^m}a_{n-m}^{m+i}.
	\end{align*}
	In particular,
	for coordinate functions $f_q^l(x)=a^l_q$,
	we have 
	\begin{align*}
	\nabla^la_q^l=\sum_i a^{l+i}e_{q-i}T^i,\quad
	\nabla^ra_q^l=\sum_i a_{-l}^{l+i}e_{l+q}T^i,
	\end{align*}
	where $e_{i}\in\mathbb{R}^\infty$ denotes the sequence with $1$ at the $i$-th entry and $0$ elsewhere.
\end{proposition}

\begin{proposition}
	The Poisson bracket for coordinate functions is
	\begin{align}\label{poisson}
		\{a_q^l,a_s^v\}=
		\frac{1}{2}a_{q}^{v+s-q}a_s^{l+q-s}\left(\sgn_{s+v-q-l}-\sgn_{q-s}\right).
	\end{align}
\end{proposition}
These results follow directly from the reduction of partial difference operators and coincide with those in \cite{refactorization2020}.

Another reduction is to sparse operators. A pseudo partial difference operator is called $(p_1,p_2)$ and $(q_1,q_2)$ sparse if its support is an arithmetic progression with steps $(p_1,p_2)$ and $(q_1,q_2)$, where $(p_1,p_2), (q_1,q_2) \in \mathbb{Z}^2$ are linearly independent. 
It is worth noting that the choice of such step vectors is not unique; rather, they serve as a $\mathbb{Z}$-basis for the underlying sublattice of the support.
\begin{example}
The operator $1+T_2^2+T_1^3$ is $(0,2)$ and $(3,0)$ sparse,
while the operator $T_1^{-1}T_2^{-1}+1+T_1T_2^{-1}+T_1T_2+T_1^2$ is $(1,1)$ and $(2,0)$ sparse.
\end{example}

Denote the set of invertible $(p_1,p_2)$ and $(q_1,q_2)$ sparse pseudo partial difference operators by $\mathrm{I}\Psi\DO(\mathbb{Z}\cdot(p_1,p_2)+\mathbb{Z}\cdot(q_1,q_2)+*)$,
which is a Lie subgroup of $\mathrm{I}\Psi\DO$.
The corresponding Lie algebra is the space $\Psi\DO(\mathbb{Z}\cdot(p_1,p_2)+\mathbb{Z}\cdot(q_1,q_2))$ of pseudo partial difference operators supported in $\mathbb{Z}\cdot(p_1,p_2)+\mathbb{Z}\cdot(q_1,q_2)$.

However it is not a Poisson-Lie subgroup of $\mathrm{I}\Psi\DO$ in general.
In the following,
we consider a simple case $N=M$,
and we define a different Poisson structure $\pi^{p,q}$ on $\mathrm{I}\Psi\mathrm{DO}_{N,N}(\mathbb{Z}\cdot(p_1,p_2)+\mathbb{Z}\cdot(q_1,q_2)+*)$,
which has the same properties as the Poisson structure on $\mathrm{I}\Psi\mathrm{DO}_{N,N}$.
When $\gcd(p_1 q_2 - p_2 q_1, N) = 1$,
there is a group isomorphism $\mathrm{I}\Psi\mathrm{DO}_{N,N}(\mathbb{Z}\cdot(p_1,p_2)+\mathbb{Z}\cdot(q_1,q_2))\simeq\mathrm{I}\Psi\mathrm{DO}_{N,N}$:
\begin{align*}
\sum_{i,j}a^{ip_1+jq_1,ip_2+jq_2}T_1^{ip_1+jq_1}T_2^{ip_2+jq_2}
\mapsto
\sum_{i,j}\tilde{a}^{i,j}S_1^iS_2^j,
\end{align*}
where $S_1:=T_1^{p_1}T_2^{p_2}$, $S_2:=T_1^{q_1}T_2^{q_2}$ and $\tilde{a}^{i,j}_{n,m}=a^{ip_1+jq_1,ip_2+jq_2}_{np_1+mq_1,np_2+mq_2}$.
\begin{lemma}\label{lem:bijection}
	Let $\mathbf{S} = \begin{pmatrix} p_1 & q_1 \\ p_2 & q_2 \end{pmatrix}$ be the step matrix associated with the sparse operators. If the determinant of $\mathbf{S}$ is coprime to the period $N$, i.e., $\gcd(p_1 q_2 - p_2 q_1, N) = 1$, then the linear index transformation
	\begin{equation}
		\begin{pmatrix} i' \\ j' \end{pmatrix} = \begin{pmatrix} p_1 & q_1 \\ p_2 & q_2 \end{pmatrix} \begin{pmatrix} i \\ j \end{pmatrix} \pmod N
	\end{equation}
	is a bijection on the discrete torus $\mathbb{Z}_N \times \mathbb{Z}_N$.
\end{lemma}

\begin{proof}
	The index transformation is a linear map on $\mathbb{Z}_N \times \mathbb{Z}_N$ defined by the matrix $\mathbf{S}$. 
	By standard matrix algebra over the ring $\mathbb{Z}_N$, $\mathbf{S}$ admits an inverse if and only if its determinant $\Delta = p_1 q_2 - p_2 q_1$ is an invertible element in $\mathbb{Z}_N$. This holds true if and only if $\gcd(\Delta, N) = 1$. 	
 Consequently, the mapping $(i, j)^T \mapsto (i', j')^T \pmod N$ is a one-to-one correspondence iff $\gcd(\Delta, N) = 1$. 
\end{proof}
According to this Lemma,
this coprimality condition ensures that the redefined coefficients $\tilde{a}_{n,m}^{i,j}$ perfectly cover the $(N,N)$-periodic grid without overlapping or missing entries, establishing the isomorphism.
Besides, the new coefficients $\tilde{a}_{n,m}^{i,j}$ are also $(N,N)$-periodic.
The Poisson structure $\pi^{p,q}$ can be defined as the pull-back of the structure $\pi$ by this isomorphism.
Note that the operators in $\mathrm{I}\Psi\mathrm{DO}_{N,N}(\mathbb{Z}\cdot(p_1,p_2)+\mathbb{Z}\cdot(q_1,q_2))\simeq\mathrm{I}\Psi\mathrm{DO}_{N,N}$ admit a similar lexicographical order.
If we require that the resulting structure is invariant under  multiplications by $T_1$, $T_2$,
we can uniquely extend $\pi^{p,q}$ to the whole group $\mathrm{I}\Psi\mathrm{DO}_{N,N}(\mathbb{Z}\cdot(p_1,p_2)+\mathbb{Z}\cdot(q_1,q_2)+*)$.
\begin{proposition}
There is a Poisson structure $\pi^{p,q}$ on the group $\mathrm{I}\Psi\mathrm{DO}_{N,N}(\mathbb{Z}\cdot(p_1,p_2)+\mathbb{Z}\cdot(q_1,q_2)+*)$.
If $\J\subset\mathbb{Z}^2$ is an arithmetic progression with common difference $(q_1,q_2)$,
then $\mathrm{IDO}_{N,N}(\J)$ is a Poisson submanifold of $\mathrm{I}\Psi\mathrm{DO}_{N,N}(\mathbb{Z}\cdot(p_1,p_2)+\mathbb{Z}\cdot(q_1,q_2)+*)$.
\end{proposition}
The properties can be proved similar to Proposition \ref{pro_sub}.
\begin{example}
Let $(p_1,p_2)=(2,0)$, $(q_1,q_2)=(1,1)$ and $\gcd(N,2)=1$.
Consider the sparse operators of the form $a+bT_1T_2$.
The Poisson bracket $\pi^{p,q}$ on such operators can be obtained from the bracket $\pi$ on operators of the form $a+bT_2$:
\begin{align*}
\{a_{n,m},b_{n,m}\}=\frac{1}{2}a_{n,m}b_{n,m},\quad
\{b_{n,m},a_{n+1,m+1}\}=\frac{1}{2}b_{n,m}a_{n+1,m+1},\quad
\forall\,n,m\in\mathbb{Z}.
\end{align*}

\end{example}
\section{Refactorization Maps}\label{sec_refac}

Having established the Poisson-Lie structure on partial difference operators, 
we now show that the pentagram-type maps on twisted $\J$-corrugated $(N,M)$-hedra  exactly manifests as  refactorization maps in this group.

\subsection{The refactorization map and integrability}
Recall  that the   pentagram-type map on $\J$-corrugated polyhedra from Section \ref{sec_j} is governed by a partition of the support $\J = \J_+ \cup \J_-$, where $\J_+$ and $\J_-$ have the same common difference. 
Algebraically, this corresponds to decomposing the partial difference operator as $\mathcal{D} = \mathcal{D}_+ + \mathcal{D}_-$, where $\mathcal{D}_{\pm} \in \NDDO(\J_{\pm})$. 
Throughout this section, we assume that there exist integers $k_+$ and $k_-$ such that 
\begin{align}\label{con}
\J_+ \subset \{k_+\} \times \mathbb{Z}, \quad \J_- \subset \{k_-\}\times\mathbb{Z}  .
\end{align}
For example,
we take $\J_+=\{(0,0),(0,1)\}$ and $\J_-=\{(1,0),(1,1)\}$.

To algebraically integrate the geometric dynamics, we translate the pentagram-type map into a refactorization problem for the operator pair $(\mathcal{D}_+, \mathcal{D}_-)$. 
In this subsection, we rigorously establish that this algebraic procedure is generically well-defined and inherently governs the integrable structure of the mapping.

\begin{theorem}\label{theorem_refac}
	
	Let $\J_\pm\subset\mathbb{Z}^2$ be a pair of non-empty disjoint finite sets with the same common difference, satisfying the condition \eqref{con}.
	Suppose that $(\D_+,\D_-)$ is an element of the space $\NDDO(\J_+)\times\NDDO(\J_-)$.
	Consider the multivalued map of that space to itself that assigns to $\D_\pm$ new pairs of partial difference operators $\tilde{\D}_\pm$ defined by
	\begin{align}\label{eq_refac}
		\tilde{\D}_+\D_-=\tilde{\D}_-\D_+.
	\end{align}
	Then the following is true.
	\begin{enumerate}
		\item
		The map $\D_\pm\mapsto\tilde{\D}_\pm$ is a generically defined single-valued transformation of the quotient $\NDDO(\J_+)\times\NDDO(\J_-)/H\tilde{\times}H$.
		
		\item 
		
		The map above is equivalent to the following refactorization dynamics
		\begin{align*}
			\mathcal{L}:=\D_-^{-1}\D_+\mapsto\tilde{\mathcal{L}}:=\D_+\D_-^{-1}.
		\end{align*}
		Equivalently,
		it has a Lax representation
		\begin{align*}
			\mathcal{L}\mapsto \D_+\mathcal{L}\D_+^{-1}.
		\end{align*}
		\item 
		The map is Poisson.
	\end{enumerate}
	
\end{theorem}
The proof naturally extends the 1D difference operator framework \cite{refactorization2020}  to doubly periodic partial difference operators by utilizing the lexicographical ordering on $\mathbb{Z}^2$.	
First,
we need the following sequence of lemmas.
\begin{lemma}\label{lemma_left}
	Let $\mathcal{D}$ and $\mathcal{D}'$ be $(N, M)$-periodic non-degenerate partial difference operators with the same support. Assume all terms in the support share a constant power $k$ for the shift operator $T_1$.
	If $\Ker \mathcal{D}' = \Ker \mathcal{D}$ on the space of Bloch solutions for generic multipliers, then there exists a non-vanishing $(N,M)$-periodic scalar bi-infinite matrix $\alpha$ such that $\mathcal{D}' = \alpha\mathcal{D}$.
\end{lemma}

\begin{proof}
	Without loss of generality,
	let $\mathcal{D} = (\sum_{j=0}^n a^j T_2^j) T_1^k$ and $\mathcal{D}' = (\sum_{j=0}^n b^j T_2^j) T_1^k$.
Factoring out the bijection $T_1^k$ implies $\Ker \hat{\mathcal{D}}' = \Ker \hat{\mathcal{D}}$ on the Bloch space, where $\hat{\mathcal{D}} = \sum_{m=0}^n a^m T_2^m$ and $\hat{\mathcal{D}}' = \sum_{m=0}^n b^m T_2^m$. By non-degeneracy, the leading coefficients satisfy $a^n_{i,j} \neq 0$. We define the $(N,M)$-periodic bi-infinite matrix $\alpha$ by $\alpha_{i,j} := b^n_{i,j}/a^n_{i,j}$.

The remainder $R = \hat{\mathcal{D}}' - \alpha \hat{\mathcal{D}}$ has order at most $n-1$ in $T_2$. 
Consider the space $\mathcal{V}_{w_1}$ of bi-infinite matrices satisfying the quasi-periodic condition $\psi_{i+N, j} = w_1 \psi_{i,j}$ for a generic $w_1 \in \mathbb{C}^*$, with no boundary conditions along $T_2$. 
Since $\hat{\mathcal{D}}$ and $\hat{\mathcal{D}}'$ share all Bloch solutions, they share the entire kernel globally on $\mathcal{V}_{w_1}$, establishing $\Ker \hat{\mathcal{D}}|_{\mathcal{V}_{w_1}} \subset \Ker R|_{\mathcal{V}_{w_1}}$.

On $\mathcal{V}_{w_1}$, since $\hat{\mathcal{D}}$ acts purely vertically without $T_1$ shifts, the system cleanly decouples into $N$ independent 1D difference equations. A unique solution is determined by specifying initial values on $n$ transverse rows across these $N$ independent columns, yielding $\dim \Ker \hat{\mathcal{D}}|_{\mathcal{V}_{w_1}} = nN$. Similarly, when $R\neq 0$, we have $\dim \Ker R|_{\mathcal{V}_{w_1}} \le (n-1)N$. The inclusion requires $nN \le (n-1)N$, yielding a strict contradiction unless $R = 0$. Thus $\hat{\mathcal{D}}' = \alpha \hat{\mathcal{D}}$, and right-multiplying by $T_1^k$ gives $\mathcal{D}' = \alpha \mathcal{D}$.
\end{proof}

\begin{lemma}\label{lemma_unique_sol}
	There exists a Zariski open and dense subset $\mathcal{U}(\J_\pm) \subset \NDDO(\J_+) \times \NDDO(\J_-)$ such that for any $(\D_+, \D_-) \in \mathcal{U}(\J_\pm)$, the equation \eqref{eq_refac}  admits a solution $(\tilde{\D}_+, \tilde{\D}_-) \in \NDDO(\J_+) \times \NDDO(\J_-)$, unique up to the left multiplication by a common non-vanishing \((N,M)\)-periodic scalar bi-infinite matrix $\alpha$.
\end{lemma}

\begin{proof}
	Let $\mathcal{U}'(\J_\pm)\subset \NDDO(\J_+) \times \NDDO(\J_-)$ be the set of $\D_\pm$ such that \eqref{eq_refac} has a unique solution \((\tilde{\D}_+,\tilde{\D}_-)\in \DO(\J_+) \times \DO(\J_-)\) up to the left multiplication by  a common non-vanishing \((N,M)\)-periodic scalar bi-infinite matrix $\alpha$.
	Let \(|\J_+| = p\) and \(|\J_-| = q\). 
	We first count the degrees of freedom and constraints in the defining equation.
	Each operator \(\tilde{\D}_+ \in \NDDO(\J_+)\) has \(p \cdot NM\) independent scalar coefficients. Similarly, \(\tilde{\D}_- \in \NDDO(\J_-)\) has \(q \cdot NM\) coefficients. The total number of unknowns is \((p+q) \cdot NM\).
	Both sides of \eqref{eq_refac} are partial difference operators supported in \(\J_+ + \J_-\).
	Since $\J_+$ and $\J_-$ have the common difference, the number of elements in the Minkowski sum \(\J_+ + \J_-\) is \(p+q-1\). 
	Therefore, the system \eqref{eq_refac} gives \((p+q-1) \cdot NM\) homogeneous linear equations in the coefficients of \(\tilde{\D}_+, \tilde{\D}_-\). 
	The difference between the number of unknowns and equations is exactly \(NM\).
	When theses equations are linear independent, this difference  corresponds to the freedom of left multiplication by an \((N,M)\)-periodic scalar bi-infinite matrix (with \(NM\) free parameters).
	Therefore,
	the uniqueness (up to the left action by \(\alpha\)) of the solution is equivalent to the condition that the corresponding \((p+q-1)\times(p+q)\) coefficient-matrix has rank $(p+q-1)$,
	so the set $\mathcal{U}'(\J_\pm)$ is Zariski open.
	
	Then we define $\mathcal{U}(\J_\pm)\subset\mathcal{U}'(\J_\pm)$ as the set of pairs $\D_\pm$ that have the property that the solution $\tilde{\D}_\pm$ of \eqref{eq_refac} belongs to $\NDDO(\J_+)\times\NDDO(\J_-)$.
	Since the coefficients of  \(\tilde{\D}_\pm\) are obtained from solving a linear system of equations,
	the operators \(\tilde{\D}_\pm\) are non-degenerate when certain rational functions do not vanish,
	which means \(\mathcal{U}(\J_\pm)\) is Zariski open in $\mathcal{U}'(\J_\pm)$ and hence in \(\NDDO(\J_+)\times\NDDO(\J_-)\).
	
	Finally, we prove the set $\mathcal{U}(\J_\pm)$ is non-empty. 
	Let $\D_\pm\in\NDDO(\J_\pm)$ be  operators with constant coefficients such that $\Ker\D_+\cap\Ker\D_-=0$.
	Since $\D_\pm$ commute with each other,
	we have $\tilde{\D}_\pm:=\D_\pm$ is a solution.
	Let $(\tilde{\mathcal{D}}'_+, \tilde{\mathcal{D}}'_-)$ be another non-degenerate solution to \eqref{eq_refac}. To rigorously compare the kernels, we restrict the operators to  space $\mathcal{V}_{w_1}$ of bi-infinite matrices satisfying the quasi-periodic condition $\psi_{i+N, j} = w_1 \psi_{i,j}$ for a generic $w_1 \in \mathbb{C}^*$, with no boundary conditions along $T_2$. 
	
	Applying \eqref{eq_refac} to any Bloch solution $\psi \in \Ker \mathcal{D}_+|_{\mathcal{V}_{w_1}}$, we obtain $\tilde{\mathcal{D}}_+ \mathcal{D}_- \psi = 0$, implying $\mathcal{D}_-(\Ker \mathcal{D}_+|_{\mathcal{V}_{w_1}}) \subset \Ker \tilde{\mathcal{D}}_+|_{\mathcal{V}_{w_1}}$. 
	Since the non-degenerate operators are restricted to $\mathcal{V}_{w_1}$, their kernels form finite-dimensional subspaces.
	Under the generic condition $\Ker \mathcal{D}_+ \cap \Ker \mathcal{D}_- = \{0\}$, $\mathcal{D}_-$ is injective on $\Ker \mathcal{D}_+$, yielding $\dim \mathcal{D}_-(\Ker \mathcal{D}_+|_{\mathcal{V}_{w_1}}) = \dim \Ker \mathcal{D}_+|_{\mathcal{V}_{w_1}}$.
	Since both are finite-dimensional subspaces, inclusion alongside equal dimensions guarantees exact equality: $\Ker \tilde{\mathcal{D}}_+|_{\mathcal{V}_{w_1}} = \mathcal{D}_-(\Ker \mathcal{D}_+|_{\mathcal{V}_{w_1}})$. 
	
	The same argument for $\tilde{\mathcal{D}}'_+$ yields $\Ker \tilde{\mathcal{D}}'_+|_{\mathcal{V}_{w_1}} = \mathcal{D}_-(\Ker \mathcal{D}_+|_{\mathcal{V}_{w_1}}) = \Ker \tilde{\mathcal{D}}_+|_{\mathcal{V}_{w_1}}$. 
	Since their kernels strictly coincide for all generic multipliers, $\Ker \tilde{\mathcal{D}}'_+ = \Ker \tilde{\mathcal{D}}_+$. 
	By Lemma \ref{lemma_left}, there exists a scalar bi-infinite matrix $\alpha$ such that $\tilde{\mathcal{D}}'_+ = \alpha \tilde{\mathcal{D}}_+$.
	Substituting this back into \eqref{eq_refac}:
	\[
	\alpha \tilde{\D}_+ \D_- = \tilde{\D}_-' \D_+ \implies \alpha \tilde{\D}_- \D_+ = \tilde{\D}_-' \D_+\implies \tilde{\D}'_-=\alpha\tilde{\D}_-.
	\]
	It yields that $\tilde{\D}_\pm$ is unique up to a left scalar multiplication.
	Therefore,
	the set $\mathcal{U}(\J_\pm)$ is Zariski open and non-empty,
	hence dense.
\end{proof}

\begin{lemma}\label{lemma_se}
The map
\begin{align}\label{eq_two}
	\NDDO(\J_+)\times\NDDO(\J_-)/H\tilde{\times}H\rightarrow
	\NDDO(\J_-)^{-1}\NDDO(\J_+)/\mathrm{Ad}\,H
\end{align}
given by 
\begin{align*}
(\D_+,D_-)\mapsto
\D_-^{-1}\D_+
\end{align*}
is generically a bijection,
where \(H\tilde{\times}H\) is the group of couples of non-vanishing $(N,M)$-quasi-periodic bi-infinite matrices with the same monodromy acting by the action \eqref{eq_lr}.
\end{lemma}
\begin{proof}
The map taking the left-right orbit of $(\mathcal{D}_+, \mathcal{D}_-)$ to the $H$-conjugacy class of $\mathcal{D}_-^{-1}\mathcal{D}_+$ is surjective by definition of the codomain. 
Injectivity generically follows from the unique factorization of rational pseudo partial difference operators.
Suppose that $\mathcal{D}_-^{-1}\mathcal{D}_+$ is $H$-conjugate to $(\mathcal{D}_-')^{-1}\mathcal{D}_+'$, that is
\begin{align}\label{eq_ta}
\mathcal{D}_-^{-1}\mathcal{D}_+=\alpha^{-1} (\mathcal{D}_-')^{-1}\mathcal{D}_+'\alpha
\end{align}
for some non-vanishing $(N,M)$-periodic bi-infinite matrix $\alpha$.
Recall that the dual of an operator $\D=\sum_{i,j}a^{i,j}T_1^iT_2^j$ is defined as
 \(\D^*=\sum_{i,j}\tilde{a}^{i,j}T_1^{-i}T_2^{-j}\),
where $\tilde{a}^{i,j}_{n,m}=a^{i,j}_{n-i,m-j}$.
Consider the refactorization equation
\begin{align}\label{eq_ad}
\tilde{\D}_+\D_-^*=\tilde{\D}_-\D_+^*,
\end{align}
with
$
\tilde{\D}_\pm\in\NDDO(-\J_\pm),
$
where $-\J:=\{-j\mid j\in\J\}$.
According to Lemma \ref{lemma_unique_sol},
equation \eqref{eq_ad} has a unique solution $\tilde{\D}_\pm$,
if $\D_\pm^*$ is an element of some open and dense subset $\mathcal{U}(-\J_\pm)\subset\NDDO(-\J_+)\times\NDDO(-\J_-)$.
Taking the dual of \eqref{eq_ad},
we have
\begin{align*}
\D_-\tilde{\D}_+^*=\D_+\tilde{\D}_-^*,
\end{align*}
and hence
\begin{align*}
\tilde{\D}_+^*(\tilde{\D}_-^*)^{-1}=\D_-^{-1}\D_+=\alpha^{-1} (\mathcal{D}_-')^{-1}\mathcal{D}_+'\alpha.
\end{align*}
where  we use \eqref{eq_ta} for the last equality.
It yields
\begin{align*}
\mathcal{D}_-'\alpha\tilde{\D}_+^*
=\mathcal{D}_+'\alpha\tilde{\D}_-^*.
\end{align*}
Now we consider a subset $\hat{\mathcal{U}}(-\J_\pm)\subset\mathcal{U}(-\J_\pm)$,
which consists of $\D_\pm^*$ such that the adjoint $\tilde{\D}_\pm^*$ of the corresponding solution $\tilde{\D}_\pm$ of \eqref{eq_ad},
is an element of $\mathcal{U}(\J_\pm)$.
This set is open and dense.
Since $\D'_\pm\alpha$ and $\D_\pm$ satisfy the same equation, by Lemma \ref{lemma_unique_sol},
we have
$\D'_\pm\alpha=\beta\D_\pm$ for some $(N,M)$-periodic bi-infinite matrix $\beta$,
which means that $\D_\pm$ and $\D_\pm'$ are in the same left-right orbit.
\end{proof}

\begin{proof}[Proof of Theorem \ref{theorem_refac}]
	
	First, we prove the single-valuedness of the map $\D_\pm\mapsto\tilde{\D}_\pm$ on the quotient.
	By Lemma \ref{lemma_unique_sol},
	for generic $\D_\pm\in\NDDO(\J_\pm)$,
	the equation \eqref{eq_refac} determines a  unique operator \(\tilde{\D}_\pm\in\NDDO(\J_\pm)\) up to left multiplication by a non-vanishing $(N,M)$-periodic scalar bi-infinite matrix.	
	If we apply the left-right action of $H \tilde{\times} H$, transforming $\mathcal{D}_{\pm}$ to $\beta^{-1}\mathcal{D}_{\pm}\alpha$, the equation \eqref{eq_refac} becomes $\tilde{\mathcal{D}}_+(\beta^{-1}\mathcal{D}_-\alpha) = \tilde{\mathcal{D}}_-(\beta^{-1}\mathcal{D}_+\alpha)$. 
	Then $\gamma^{-1}\tilde{\mathcal{D}}_{\pm}\beta$ is a solution,
	where $\alpha$, $\beta$, $\gamma$ have the same monodromy.
	Therefore, the map sends elements of the same left-right orbit to elements of the same left-right orbit, descending to a well-defined single-valued transformation on the quotient $\NDDO(\J_+)\times\NDDO(\J_-)/H\tilde{\times}H$.
	
	Then we prove the second statement.	According to Lemma \ref{lemma_se},
	there is a bijection \eqref{eq_two}.	
	From the refactorization relation $\tilde{\mathcal{D}}_+\mathcal{D}_- = \tilde{\mathcal{D}}_-\mathcal{D}_+$, multiplying by $\tilde{\mathcal{D}}_-^{-1}$ on the left and $\mathcal{D}_-^{-1}$ on the right yields $\tilde{\mathcal{D}}_-^{-1}\tilde{\mathcal{D}}_+ = \mathcal{D}_+\mathcal{D}_-^{-1}$. Defining the Lax operators $\mathcal{L} := \mathcal{D}_-^{-1}\mathcal{D}_+$ and $\tilde{\mathcal{L}} := \tilde{\mathcal{D}}_-^{-1}\tilde{\mathcal{D}}_+$, we directly obtain $\tilde{\mathcal{L}} = \mathcal{D}_+\mathcal{D}_-^{-1} = \mathcal{D}_+(\mathcal{D}_-^{-1}\mathcal{D}_+)\mathcal{D}_+^{-1} = \mathcal{D}_+\mathcal{L}\mathcal{D}_+^{-1}$. This gives the required Lax representation.
	
	Finally,
	we consider the third statement  (Poisson property).
	The proof is similar to the one-dimensional pseudo-difference operators \cite{refactorization2020},
	so we briefly give the process.
	The map can be decomposed into some operations within the Poisson-Lie group $\mathrm{I}\Psi \DO$
	\begin{align*}
	(\D_+,\D_-)\rightarrow
	\D_+\D_-^{-1}=\tilde{\D}_-^{-1}\tilde{\D}_+\rightarrow(\tilde{\D}_+,\tilde{\D}_-).
	\end{align*}
	The first arrow is well-defined according to Lemma \ref{lemma_unique_sol}.	
	Besides, this map is Poisson, since multiplication in the group of pseudo partial difference operators is Poisson,
	and inversion is anti-Poisson. 
	The second arrow is well-defined since there is a generic  bijection between $\NDDO(\J_+)\times\NDDO(\J_-)/H\tilde{\times}H$ and $
	\NDDO(\J_-)^{-1}\NDDO(\J_+)/\mathrm{Ad}\,H$ by Lemma \ref{lemma_se}.
	Therefore, the
	whole map is a composition of Poisson maps and hence Poisson.
\end{proof}

For the sake of concreteness, in what follows we focus on the example defined by $\J_+ = \{(0,0), (0,1)\}$ and $\J_- = \{(1,0), (1,1)\}$.
\begin{proposition}
	For this specific choice of $\J_\pm$,
	the so-obtained map $\D_\pm\mapsto\tilde{\D}_\pm$ can be extended to a map
	\begin{align*}
		(\D_\pm,\sum_{k=1}^4 Q_k)\longmapsto
		(\tilde{\D}_\pm,\sum_{k=1}^4 Q_k)
	\end{align*}  
	where $\sum_{k=1}^4 Q_k$ is a real divisor consisting of four distinct points on the associated smooth spectral curve $\Gamma_\D$,
	considered modulo the induced scaling action \eqref{eq_act_w}.
	Besides,
	this map coincides with the pentagram-type map \eqref{def_lap} on $\mathcal{P}_{N,M}^3(\J)$,
	which is the space of generic twisted $\J$-corrugated $(N,M)$-hedra in $\mathbb{RP}^3$ modulo projective transformations.
\end{proposition}

\begin{proof}
	According to Proposition \ref{pro_one-to-one},
	given a generic twisted $\J$-corrugated $(N,M)$-hedron $\{v_{i,j}\}$ in $\mathbb{RP}^3$ modulo projective transformations,
	there is a unique equivalence class of $(\D,\sum_{k=1}^4 Q_k)$,
	where $\D\in\NDDO(\J)/H\tilde{\times}H$ and $\sum_{k=1}^4 Q_k$ is a real divisor consisting of four distinct points on $\Gamma_\D$ considered up to the action \eqref{eq_act_w}.
	The partial difference operator $\D$ admits a unique factorization
	$
	\D=\D_++\D_-,
	$
	where $\D_\pm\in\NDDO(\J_\pm)$ and similar for $\tilde{\D}=\tilde{\D}_++\tilde{\D}_-$.

	Suppose that $V$ is a lift of $v$.
	Applying both sides of \eqref{eq_refac} to $V$,
	we get
	\begin{align*}
		\tilde{\D}_+\D_-V=\tilde{\D}_-\D_+V,
	\end{align*}
	which, using that $\D V=0$ and thus $\D_-V=-\D_+V$,
	can be rewritten as
	\begin{align*}
		\tilde{\D}\D_+V=0.
	\end{align*}
	This means that $\tilde{V}=\D_+V$. 
	Note that the vector $(\D_+V)_{i,j}$ belongs to the line passing through $V_{i,j}$, $V_{i,j+1}$.
	At the same time,
	we have $\tilde{V}=\D_+V=-\D_-V$,
	so
	\begin{align*}
		(\D_+V)_{i,j}=-(\D_-V)_{i,j}\in\,\operatorname{span}(V_{i+1,j},V_{i+1,j+1}).
	\end{align*}
	Thus we have
	\begin{align*}
		\tilde{V}_{i,j}\in\operatorname{span}(V_{i,j},V_{i,j+1})\cap\operatorname{span}(V_{i+1,j},V_{i+1,j+1}),
	\end{align*}
	which coincides with \eqref{def_lap}.
	The new twisted $\J$-corrugated $(N,M)$-hedron $\{\tilde{v}_{i,j}\}$ also corresponds to a divisor.
	
	Now we prove the divisor $\sum_{k=1}^4 Q_k$ is invariant under the map \eqref{def_lap}.
	Note that the four components of $V$ coincides with four Bloch solutions with monodromies $Q_1,Q_2,Q_3,Q_4$ on $\Gamma_\D$, that is
	\begin{align*}
		V=K\begin{pmatrix}
			\psi(Q_1)\\
			\psi(Q_2)\\
			\psi(Q_3)\\
			\psi(Q_4)
		\end{pmatrix},
	\end{align*}
	where $K\in\text{GL}(4,\mathbb{R})$ is a constant matrix and $\psi(Q_k)$, $k=1,\ldots,4$, are the Bloch solutions of $\D$ at $Q_k\in\Gamma_\D$.
	Let $Q_k=(Q_{k,1},Q_{k,2})$, and
	we can prove that the monodromies of $V$ are
	\begin{align*}
		\phi_1=K\operatorname{diag}\bigl(Q_{1,1}, Q_{2,1},Q_{3,1},Q_{4,1})K^{-1},\quad
		\phi_2=K\operatorname{diag}\bigl(Q_{1,2}, Q_{2,2},Q_{3,2},Q_{4,2})K^{-1}.
	\end{align*}
	Define $\tilde{\psi}_k:=\D_+\psi(Q_k)$, and obviously it satisfies $\tilde{\D}\tilde{\psi}_k=0$.
	Since $\D_+$ is $(N,M)$-periodic,
	$\psi(Q_k)$ and $\tilde{\psi}_k$ admit the same monodromies, i.e.
	\begin{align*}
		T_1^N\tilde{\psi}_k=Q_{k,1}\tilde{\psi}_k,\quad
		T_2^M\tilde{\psi}_k=Q_{k,2}\tilde{\psi}_k,\quad k=1,2,3,4.
	\end{align*}
	Therefore,
	we have $\tilde{\psi}_k=\tilde{\psi}(Q_k)$.
	In fact, the spectral curve $\Gamma_{\tilde{\D}}$ is the same as $\Gamma_\D$.
	Finally, $\tilde{V}$ has the form
	\begin{align*}
		\tilde{V}=\D_+V=K\D_+\begin{pmatrix}
			\psi(Q_1)\\
			\psi(Q_2)\\
			\psi(Q_3)\\
			\psi(Q_4)
		\end{pmatrix}
		=K\begin{pmatrix}
			\tilde{\psi}(Q_1)\\
			\tilde{\psi}(Q_2)\\
			\tilde{\psi}(Q_3)\\
			\tilde{\psi}(Q_4)
		\end{pmatrix},
	\end{align*} 
	where $\D_+ K=K\D_+$ because $K$ is a constant matrix.
	Therefore,
	the new twisted $\J$-corrugated $(N,M)$-hedron $\tilde{V}$ corresponds to the same divisor $\sum_{k=1}^4 Q_k$.
\end{proof}

\subsection{Explicit formulas for the pentagram-type map}
To provide a concrete computational framework, we now focus on the explicit formulas for the pentagram-type map \eqref{def_lap} on twisted $\J$-corrugated $(N,M)$-hedra.
\begin{proposition}\label{candi_D}
Under the action \eqref{eq_lr} of the group $H\tilde{\times}H$, 
every partial difference operator $\mathcal{D}_0=a+bT_2+cT_1+dT_1T_2 \in \NDDO(\J)$ can be uniquely reduced to the representative form
	\begin{align}\label{expr_D}
		\mathcal{D}=A+T_2+BT_1+T_1T_2.
	\end{align}
up to a residual action
\begin{align}\label{AB_action}
	(A_{i,j},B_{i,j})\rightarrow\left(\frac{\gamma_j}{\gamma_{j+1}}A_{i,j}, \frac{\gamma_j}{\gamma_{j+1}}B_{i,j}\right),
\end{align}
by an arbitrary non-vanishing $M$-quasi-periodic sequence $\gamma_j$ ($\gamma_{j+M}=\mu\gamma_j$).
The corresponding $(N,M)$-periodic coordinates are explicitly defined as:
	\begin{align}\label{eq_AB}
		A_{i,j}:=\frac{a_{i,j}}{b_{i,j}}\xi_{i,j}, \quad B_{i,j}:=\frac{c_{i,j}}{d_{i,j}}\xi_{i+1,j},
	\quad
	\text{where}
	\quad
		\xi_{i,j}:=\prod_{k=0}^{i-1}\frac{b_{k,j-1}d_{k,j}}{d_{k,j-1}b_{k,j}}.
	\end{align}
\end{proposition}

\begin{proof}
	Let $(\alpha,\,\beta)\in H\tilde{\times}H$  be two non-vanishing bi-infinite matrices with the same monodromy, 
	and we consider their action on the operator $\D_0$ 
	\begin{align}\label{ope}
		\D_0\rightarrow\beta^{-1}\D_0\alpha=
		\beta^{-1} a\alpha+\beta^{-1} b\alpha_2T_2+\beta^{-1} c\alpha_1T_1+\beta^{-1} d\alpha_{12}T_1T_2.
	\end{align}
	Suppose that the coefficients of $T_2$ and $T_1T_2$ are $1$, i.e.
	\begin{align*}
		\beta^{-1}b\alpha_2=1,\quad
		\beta^{-1} d\alpha_{12}=1.
	\end{align*}
	The solution for $\alpha$ and $\beta$ is
	\begin{align}\label{eq_alpha}
		\alpha_{i,j}=\alpha_{0,j}\prod_{k=0}^{i-1}\frac{b_{k,j-1}}{d_{k,j-1}},\quad
		\beta_{i,j}=b_{i,j}\alpha_{0,j+1}\prod_{k=0}^{i-1}\frac{b_{k,j}}{d_{k,j}},
	\end{align}
	where $\alpha_{0,j}$, $j\in\mathbb{Z}$ are undetermined nonzero constants.
	The remaining coefficients of $1$ and $T_1 $ in \eqref{ope} are 
	\begin{align}\label{eq_beta}
		\beta^{-1} a\alpha=\frac{\alpha_{0,j}}{\alpha_{0,j+1}}\frac{a_{i,j}}{b_{i,j}}\xi_{i,j},\quad
		\beta^{-1} c\alpha_1=\frac{\alpha_{0,j}}{\alpha_{0,j+1}}\frac{c_{i,j}}{d_{i,j}}\xi_{i+1,j},\quad
		\xi_{i,j}:=\prod_{k=0}^{i-1}\frac{b_{k,j-1}d_{k,j}}{d_{k,j-1}b_{k,j}}.
	\end{align}
	From \eqref{eq_alpha},
	we have $\alpha_{i+N,j}=\alpha_{i,j}$, $\alpha_{i,j+M}=\frac{\alpha_{0,j+M}}{\alpha_{0,j}}\alpha_{i,j}$ and $\beta_{i+N,j}=\beta_{i+N,j}$, $\beta_{i,j+M}=\frac{\alpha_{0,j+M+1}}{\alpha_{0,j+1}}\beta_{i,j}$.
	To ensure  $(\alpha,\beta)\in H\tilde{\times}H$, 
	we need the condition $\frac{\alpha_{0,j+M}}{\alpha_{0,j}}=\frac{\alpha_{0,j+M+1}}{\alpha_{0,j+1}}$.
	Equivalently, 
	we have $\alpha_{0,j}=\gamma_{j}$ is a quasi-periodic sequence, that is, $\gamma_{j+M}=\mu\gamma_{j}$ for a fixed $\mu\in\mathbb{R}^*$.
	The coefficients \eqref{eq_beta} reduces to \eqref{eq_AB} up to the action \eqref{AB_action}.
\end{proof}

\begin{proposition}
	The pentagram-type map \eqref{def_lap} takes the following form
	\begin{align}\label{eq_main}
		\tilde{A}_{i,j}=\frac{A_{i+1,j+1}-B_{i,j+1}}{A_{i+1,j}-B_{i,j}}A_{i+1,j},\quad
		\tilde{B}_{i,j}=\frac{A_{i+1,j+1}-B_{i,j+1}}{A_{i+1,j}-B_{i,j}}B_{i,j},
	\end{align}
	where $(A_{i,j},B_{i,j})$ and $(\tilde{A}_{i,j},\tilde{B}_{i,j})$ are unique up to the action \eqref{AB_action} for (different) non-vanishing $M$-quasi-periodic sequences.
\end{proposition}
\begin{proof}
	Let $\J = \{(0,0), (0,1), (1,0), (1,1)\}$.
	Suppose that  $(v_{i,j})\in\mathcal{P}_{N,M}^3(\J)$ is a generic twisted $\J$-corrugated $(N,M)$-hedron in $\mathbb{RP}^3$.
	According to Proposition \ref{pro_one-to-one},
	there is a partial difference operator $\D\in\NDDO(\J)$ corresponding to this polyhedron.
	According to Proposition \ref{candi_D}, $\D$ has the form \eqref{expr_D}.
	Let $V$ be a lift of $v$ to $\mathbb{R}^4$ such that $\D V=0$, that is
	\begin{align}\label{eq_recur}
		A_{i,j}V_{i,j}+V_{i,j+1}+B_{i,j}V_{i+1,j}+V_{i+1,j+1}=0,
	\end{align}
	where $A$, $B$ are unique up to the action \eqref{AB_action}.
	According to the definition of the pentagram-type map \eqref{def_lap},
	we have 
	\begin{align}\label{eq_evo}
		\tilde{V}_{i,j}=\lambda_{i,j}(A_{i,j}V_{i,j}+V_{i,j+1}),
	\end{align}
	where $\lambda_{i,j}$ is undetermined.
	The compatibility condition of \eqref{eq_recur} and \eqref{eq_evo} gives  
	\begin{align*}
		&\lambda_{i,j}\tilde{A}_{i,j}(B_{i,j}V_{i+1,j}+V_{i+1,j+1})+
		\lambda_{i,j+1}(B_{i,j+1}V_{i+1,j+1}+V_{i+1,j+2})\\
		=&
		\lambda_{i+1,j}\tilde{B}_{i,j}(A_{i+1,j}V_{i+1,j}+V_{i+1,j+1})
		+\lambda_{i+1,j+1}(A_{i+1,j+1}V_{i+1,j+1}+V_{i+1,j+2}).
	\end{align*}
	Since $V_{i+1,j}$, $V_{i+1,j+1}$ and $V_{i+1,j+2}$ are generically linearly independent,
	their coefficients give us
	\begin{align}\label{eq_aa}
		\tilde{A}_{i,j}=\frac{\lambda_{j+1}}{\lambda_{j}}\frac{A_{i+1,j+1}-B_{i,j+1}}{A_{i+1,j}-B_{i,j}}A_{i+1,j},\quad
		\tilde{B}_{i,j}=\frac{\lambda_{j+1}}{\lambda_{j}}\frac{A_{i+1,j+1}-B_{i,j+1}}{A_{i+1,j}-B_{i,j}}B_{i,j},
	\end{align}
	and $\lambda_{i,j}$ is proved to be independent of $i$.
	To ensure $(\tilde{A},\tilde{B})$ is $(N,M)$-periodic,
	we need the condition that $\lambda_j$ is $M$-quasi-periodic.
	If we consider the different coefficients $(\frac{\gamma_j}{\gamma_{j+1}}A_{i,j},\frac{\gamma_j}{\gamma_{j+1}}\gamma_jB_{i,j})$ in \eqref{eq_recur},
	then \eqref{eq_aa} becomes 
	\begin{align*}
		\tilde{A}_{i,j}=\frac{\lambda_{j+1}\gamma_{j+1}}{\lambda_{j}\gamma_{j+2}}\frac{A_{i+1,j+1}-B_{i,j+1}}{A_{i+1,j}-B_{i,j}}A_{i+1,j},\quad
		\tilde{B}_{i,j}=\frac{\lambda_{j+1}\gamma_{j+1}}{\lambda_{j}\gamma_{j+2}}\frac{A_{i+1,j+1}-B_{i,j+1}}{A_{i+1,j}-B_{i,j}}B_{i,j}.
	\end{align*}
	Since $\frac{\gamma_{j+1}}{\lambda_j}$ is also $M$-quasi-periodic,
	this proposition is proved.
\end{proof}

\begin{proposition}\label{pro_bracket}
In the coordinates $\{A_{i,j},B_{i,j}\}$,
the Poisson structure takes the following form:
\begin{align*}
\{A_{i,j},A_{i,j+1}\}=-\frac{1}{2}A_{i,j}A_{i,j+1},\quad
\{B_{i,j},B_{i,j+1}\}=\frac{1}{2}B_{i,j}B_{i,j+1},\\
\{A_{i+n,j},A_{i,j}\}=A_{i+n,j}A_{i,j},\quad
\{B_{i+n,j},B_{i,j}\}=B_{i+n,j}B_{i,j},\\
\{A_{i+n,j},A_{i,j\pm 1}\}=-\frac{1}{2}A_{i+n,j}A_{i,j\pm 1},\quad
\{B_{i+n,j},B_{i,j\pm 1}\}=-\frac{1}{2}B_{i+n,j}B_{i,j\pm 1},
\end{align*}
and 
\begin{align*}
\{A_{i,j},B_{i,j}\}=-A_{i,j}B_{i,j},\quad
\{A_{i,j\pm 1},B_{i,j}\}=\frac{1}{2}A_{i,j\pm 1}B_{i,j},\\
\{A_{i+n,j},B_{i,j}\}=A_{i+n,j}B_{i,j},\quad
\{A_{i+n,j},B_{i,j\pm 1}\}=-\frac{1}{2}A_{i+n,j}B_{i,j\pm 1},\\
\{B_{i+n,j},A_{i,j}\}=B_{i+n,j}A_{i,j},\quad
\{B_{i+n,j},A_{i,j\pm 1}\}=-\frac{1}{2}B_{i+n,j}A_{i,j\pm 1},
\end{align*}
where $i=1,\ldots,N$, $j=1,\ldots,M$, and $n=1,\ldots,N-i$.
\end{proposition}
\begin{proof}
The Poisson bracket of $\D$ is the product bracket corresponding to the decomposition $\D=\D_++\D_-$,
where $\D_+=a+bT_2,$ $\D_-=cT_1+dT_1T_2.$
The bracket on $\D_+$ is defined using \eqref{Poisson_coordinate},
while the $\D_-$ part is endowed with the negative of that bracket.
Explicitly,
we have
\begin{align*}
	\{a_{i,j},b_{i,j}\}=\frac{1}{2}a_{i,j}b_{i,j},\quad
	\{b_{i,j},a_{i,j+1}\}=\frac{1}{2}b_{i,j}a_{i,j+1},
\end{align*}
and
\begin{align*}
	\{c_{i,j},d_{i,j}\}=-\frac{1}{2}c_{i,j}d_{i,j},\quad
	\{d_{i,j},c_{i,j+1}\}=-\frac{1}{2}d_{i,j}c_{i,j+1}.
\end{align*}
The brackets of functions \eqref{eq_AB} can be computed using these formulas by a straightforward calculation.
Here the coordinates $\{A_{i,j}, B_{i,j}\}$ are considered up to the action \eqref{AB_action},
which preserves the Poisson brackets.
Besides, the brackets are invariant under the map \eqref{eq_main},
which can be verified using a computer algebra system.
\end{proof}

\subsection{Relation to the Y-system}
Recall that given four coplanar vectors $X,Y,Z,W$ such that
\begin{align*}
	Y=\lambda_1X+\lambda_2Z,\quad
	W=\mu_1 X+\mu_2 Z,
\end{align*}
where $\lambda_i$, $\mu_i$, $i=1,2$ are constants,
the cross-ratio of the lines spanned by these vectors is given by
\begin{align*}
	[X,Y,Z,W]=\frac{\lambda_2\mu_1}{\lambda_1\mu_2}.
\end{align*}
If we introduce the Y-variable:
\begin{align*}
y_{i,j}=-[v_{i+1,j},\tilde{v}_{i,j},v_{i+1,j+1},\tilde{v}_{i+1,j}]=-\frac{A_{i+1,j}}{B_{i,j}},
\end{align*}
then the equation \eqref{eq_main} turns into the following Y-system 
\begin{align}\label{eq_y}
y_{i+2,j+1}\tilde{\tilde{y}}_{i+1,j}=\frac{(1+\tilde{y}_{i+1,j})(1+\tilde{y}_{i+2,j+1})}{\left(1+\tilde{y}_{i+2,j}^{-1}\right)\left(1+\tilde{y}_{i+1,j+1}^{-1}\right)},
\end{align}
where   $\tilde{\tilde{y}}$ denotes the variable $y$ after two consecutive iterations of the pentagram-type map \eqref{def_lap}.
It is worth noting that while the coordinates $A_{i,j}$ and $B_{i,j}$ are defined only up to the residual scaling action \eqref{AB_action}, 
the $y$-variables are uniquely determined. 
This uniqueness stems from their geometric definition as cross-ratios, which are projectively invariant.

\section{Conclusion and discussions}\label{sec_con}

In this paper, we introduced a family of pentagram-type maps on twisted $\J$-corrugated $(N,M)$-hedra in $\mathbb{RP}^3$. By establishing a canonical correspondence between the projective equivalence classes of these polyhedra and the spectral data of doubly periodic partial difference operators modulo the gauge actions, we  proved their integrability. 
We demonstrated that the geometric dynamics manifest exactly as refactorization maps on the Poisson-Lie group of pseudo partial difference operators. This algebraic identification naturally yields an explicit Lax representation, an $r$-matrix induced Sklyanin Poisson bracket, and we also find the geometric evolution can be treated as a canonical Y-system.

Our framework suggests several natural directions for future research.
The generalized pentagram maps related to Q-nets reductions have been  studied in \cite{q_net2026}, 
so it is an interesting question whether the  integrable dynamics of Q-nets can be governed by the refactorization of partial difference operators.

Furthermore, a promising direction is to connect our algebraic refactorization maps with the higher-dimensional transformation theory of  Q-nets developed by Doliwa, Santini, and Ma\~nas \cite{doliwa2000transformations}. 
Their work demonstrates that the fundamental transformation (algebraically equivalent to the binary Darboux transformation) acts on a quadrilateral strip, effectively generating an additional lattice dimension. An open challenge is to embed the refactorization dynamics $\tilde{\mathcal{D}}_+ \mathcal{D}_- = \tilde{\mathcal{D}}_- \mathcal{D}_+$ into this higher-dimensional sequence of  transformations, unifying the Poisson-Lie algebraic framework with coordinate-free discrete geometry.

Given that the coordinate evolution \eqref{eq_y} translates cleanly into a canonical Y-system, investigating the underlying cluster algebraic structures and their potential links to dimer models on doubly periodic graphs is  desirable.
	
Exploring the  continuum limits of these 2D discrete pentagram-type maps could reveal deep connections to classical integrable PDEs, such as variations of the 2D Toda lattice or the Kadomtsev-Petviashvili (KP) hierarchy.

\section*{Acknowledgement}
This work was supported by  the National Natural Science Foundation of China (Grant Nos.  12571265, 12201325, 12235007),
and Zhejiang Provincial Natural Science Foundation
of China (Grant No. LMS26A010008).

\end{document}